\documentclass[prodmode,permissions]{acmsmall-ec16-full}

\usepackage[ruled]{algorithm2e}

\usepackage{amsmath,amsfonts,amssymb,bbm}
\usepackage[numbers,sort&compress]{natbib} 
\SetArgSty{textrm}  
\SetAlFnt{\small}
\SetAlCapFnt{\small}
\SetAlCapNameFnt{\small}
\SetAlCapHSkip{0pt}
\IncMargin{-\parindent}

\usepackage{cleveref}
\usepackage{dsfont}
\usepackage{url}
\usepackage{mathrsfs}
\usepackage{framed}
\usepackage{mathtools}

\def\1{\mathds 1}

\def\R{\mathbf R}

\def\cN{\mathcal N}
\def\cO{\mathcal O}
\def\cW{\mathcal W}

\def\scrD{\mathscr D}

\def\conv{\mathrm{conv}}
\def\SPAN{\mathrm{span}}

\def\id{\mathrm{id}}

\newcommand{\spoly}[1]{G_{\mathrm{p}({#1})}}

\newcommand\blue[1]{{\color{blue} #1}}

\newcommand{\ds}{\displaystyle}
\newcommand{\bR}{\mathbf{R}}

\newcommand{\trans}{\top}
\newcommand{\N}{[n]}
\newcommand{\M}{[m]}
\newcommand{\norm}[1]{\|#1\|}
\newcommand{\cset}{H}
\newcommand{\pset}{F}
\newcommand{\Gid}{G_{\mathrm{id}}}
\newcommand{\Gpoly}[1]{G_{\mathrm{p}(#1)}}

\newcommand{\sgf}{\delta\!f}
\newcommand{\Exp}{\mathbf{E}}
\newcommand{\optq}{\bar{q}}
\newcommand{\optp}{\bar{p}}
\newcommand{\optw}{\bar{w}}
\newcommand{\subdiff}{\partial}

\newcommand{\meanv}{\tilde{v}}
\newcommand{\tk}{{\textup{k}}}
\newcommand{\etao}{\eta_{\textup{o}}}
\newcommand{\etae}{\eta_{\textup{e}}}

\DeclareMathOperator*{\argmax}{arg\,max}

\doi{XXXXXXX.XXXXXXX}

\begin{document}


\title{Rate of Price Discovery in Iterative Combinatorial Auctions}
\author{JACOB ABERNETHY
\affil{University of Michigan}
S\'EBASTIEN LAHAIE
\affil{Microsoft Research}
MATUS TELGARSKY
\affil{University of Michigan}}

\begin{abstract}
  We study a class of iterative combinatorial auctions which can be
  viewed as subgradient descent methods for the problem of pricing
  bundles to balance supply and demand. We provide concrete
  convergence rates for auctions in this class, bounding the number of
  auction rounds needed to reach clearing prices. Our analysis allows
  for a variety of pricing schemes, including item, bundle, and
  polynomial pricing, and the respective convergence rates confirm
  that more expressive pricing schemes come at the cost of slower
  convergence. We consider two models of bidder behavior. In the first
  model, bidders behave stochastically according to a random utility
  model, which includes standard best-response bidding as a special
  case. In the second model, bidders can behave arbitrarily (even
  adversarially), and meaningful convergence relies on properly
  designed activity rules.
\end{abstract}

\begin{bottomstuff}
Author addresses:
J. Abernethy, Department of Computer Science,
University of Michigan; email: \url{jabernet@umich.edu};
S. Lahaie, Microsoft Research, New York City; email:
\url{slahaie@microsoft.com};
M. Telgarsky, Department of Computer Science,
University of Michigan; email: \url{mtelgars@umich.edu}
\end{bottomstuff}

\maketitle

\section{Introduction.}

Combinatorial auctions are used to sell multiple distinct items at
once in situations where items may be substitutes or complements.
Because of their generality as a resource allocation mechanism,
combinatorial auctions have been proposed for a variety of domains
including the allocation of wireless spectrum, airport landing slots,
and real estate~\citep{cramton2006combinatorial}. In an
\emph{iterative} combinatorial auction, bidders place bids on bundles
of items in response to prices updated by the auctioneer, and this
process repeats until bidding reaches quiescence. Iterative designs
are particularly attractive for combinatorial auctions because the
sheer size of the bundle space makes it impractical to report
valuation information in a single shot.

At the core of an iterative auction design is the choice of pricing
scheme, because prices drive the information revelation process. The
current space of auctions offers two extremes: linear (item) pricing,
where the price of a bundle is the sum of its item prices, and bundle
pricing, where each bundle is explicitly priced. The two schemes have
different advantages. Linear prices provide information about final
costs even for bundles not explicitly bid on, leading to fewer rounds
of bidding---in this sense, they provide effective `price discovery'.
However, in the presence of complementarities linear prices cannot
effectively balance supply and demand, leading to inefficiencies in
the final allocation. Bundle prices, on the other hand, can always
clear the market and support an efficient allocation, but provide
limited price discovery. As a result, bundle-price auctions typically
require far more rounds in practice to reach termination, as confirmed
empirically in both simulation studies and lab
experiments~\citep{scheffel2011experimental,schneider2010robustness}.

In this paper we consider the design of iterative auctions from an
algorithmic perspective, leading to a formal study of the relationship
between price structure and price discovery. Prior work has shown that
iterative multi-item auctions generally fall under two design
paradigms: they can be viewed as either primal-dual or subgradient
algorithms to solve the dual problems of allocation and pricing, once
these are formulated as linear programming
problems~\cite{bikhchandani2002linear}. In this work we consider
auctions based on the subgradient method. We obtain general linear and
quadratic programming formulations of the pricing problem solved by
the auction, which allows us to study various price structures under a
single framework. We provide a general specification of subgradient
auctions which subsumes several prior auctions using linear or bundle
prices, and also leads to novel auction designs using polynomial
prices in between these two schemes. Polynomial prices extend linear
prices by assigning coefficients to combinations of items, and can
allow the auctioneer to strike a more careful balance between market
clearing and price discovery.

Our main focus is the convergence rate of iterative auctions,
and in particular its dependence on the pricing scheme used.
We consider two agent models. In
the first model, agents bid on their most preferred bundles at the
current prices, but their value estimates for bundles are subject to
stochastic errors at each round. This translates into stochastic behavior where
bids may be placed on lesser-preferred bundles by chance. Our
convergence analysis confirms that subgradient auctions reach clearing
prices even under this kind of imperfect bidding behavior, and bounds
the rates for linear, bundle, and polynomial pricing.
In the second model, agents can place bids arbitrarily or even
adversarially in each round upon seeing the current prices. While it
is still possible to prove certain technical convergence statements
under such a model, it unsurprisingly admits behavior that would be
unreasonable and disallowed in practice. We therefore show how the
convergence statements become more meaningful when imposing
revealed-preference activity rules (constraints on bids over rounds)
that have been proposed in
practice~\cite{ausubel2014practical}.\footnote{Under both the
  stochastic model and the adversarial model with activity rules,
  agents bid consistently with some underlying valuation for the
  bundles, but it need not correspond to their true valuation. Agents
  may benefit from strategic bidding. This has no bearing on the
  auction's worst-case convergence rate, so we set aside incentive
  concerns in this work. As with other iterative auctions used in
  practice, supplementary pricing stages such as core-selecting or VCG
  pricing may be used to achieve reasonable outcomes in
  equilibrium~\cite{day2012quadratic}.}

At the core of our analysis are tools and concepts from statistical
learning theory, online (sequential) learning, and convex
analysis. The auction is cast as an iterative procedure whose goal is
to optimize an objective function over prices, in analogy to fitting a
learning model to minimize prediction loss. The tension between market
clearing and price discovery mirrors the bias-variance trade-off
familiar from statistics. Linear prices offer `low variance', which
makes them informative about final costs even after a limited number
of rounds. Bundle prices offer `low bias', which is needed to flexibly
price bundles to balance demand and supply. (The notions of bias and
variance here are analogies, because there is no randomness in the
auction mechanism.) We emphasize the price discovery aspects
throughout the paper. While our analysis does provide structural
insights into the market clearing properties of various pricing
schemes, a detailed treatment of this question requires one to
consider restricted valuation classes, which we defer to a separate
study.

%

The remainder of the paper is organized as follows.
Section~\ref{sec:model} introduces the elements of our model and
describes our approach to price representation. Section~\ref{sec:opt}
formulates the core problems of allocation and pricing, and develops
the duality relationship between the two.
Section~\ref{sec:iter-auction} specifies the class of auctions based
on the subgradient method. The main results of the paper appear in
Sections~\ref{sec:stochastic} and~\ref{sec:arbitrary}, providing
convergence rates for the auctions under the stochastic and arbitary
bidding models, respectively. Section~\ref{sec:conclusion} concludes.
Detailed proofs of all results are deferred to the appendices.
Sketches of the proofs are given in the main text.



\paragraph{Related work} There is an extensive literature on the
design of multi-item auctions~\citep[e.g.][]{milgrom2004putting}. Our
paper relates to the narrower literature on the algorithmic properties
of iterative auctions. The allocation problem at the basis of
combinatorial auctions was first formulated in a linear programming
(LP) framework by~\citet{bikhchandani2002package}, who provide
formulations leading to linear and bundle prices, both anonymous and
personalized. The formulation used in this paper is equivalent to the
one used by~\citet{lahaie2011kernel} for his primal-dual auction, and
subsumes the LPs of~\citet{bikhchandani2002package}. The work
of~\citet{lahaie2009kernel,lahaie2011kernel} introduced the use of
polynomial prices for combinatorial auctions.

Drawing on the LP viewpoint,~\citet{bikhchandani2002linear}
and~\citet{de2007ascending} categorized existing auctions as either
primal-dual or subgradient algorithms. Subgradient auctions include
the uniform price auction for homogeneous goods, and the well-known
auctions of~\citet{crawford1981job}, \citet{kelso1982job},
\citet{parkes1999bundle}, and~\citet{ausubel2002ascending} for
heterogeneous goods. All of these auctions are part of the class
studied here (for a suitable choice of step-size policy).

An important aspect of our auction class is that we do not enforce
price monotonicity: the price of a bundle may ascend and descend
during the course of the auction. While there is precedent for
non-monotone price paths~\citep[e.g.][]{ausubel2006efficient}, most
auction designs favor ascending prices to rule out certain gaming
behaviors on the part of bidders that aim to delay termination. We
show how activity rules can mitigate such concerns and restore
convergence. \citet{ausubel2014practical} discuss the practical
aspects of combinatorial auction design,
while~\citet{ausubel2011activity} consider activity rules.
\citet{harsha2010strong} provide a detailed treatment of revealed
preference activity rules similar to the rule considered in this
paper.

Many of the core ideas used herein draw from convex
analysis~\cite{borwein2010convex,HULL,ROC} and the mathematical
results that explore sequential optimization procedures such as
subgradient descent, or more broadly \emph{mirror
  descent}~\cite{beck2003mirror}. Many of these methods have recently
been employed within the machine learning community, where prediction
and decision problems are frequently viewed as an \emph{online convex
  optimization game}~\citep{zinkevich2003online}. In this model, a
learner is repeatedly asked to select decisions from a convex set, and
on each round the feedback on this decision arrives in the form of a
convex loss function, potentially selected by an adversary. In the
context of iterative auctions, the decision is the price vector at
each round, and the resulting bids provide feedback on the `loss'
incurred. Online convex optimization games have been thoroughly
explored in recent years and one can find a number of comprehensive
surveys, including the work
of~\citet{cesa2006prediction},~\citet{hazan2012survey},
and~\citet{shalev2011online}.


\section{The formal model.}
\label{sec:model}

We consider a model with $n$ agents (buyers) and a single seller
holding $m$ distinct items. At a high level, the purpose of the
auction is to allocate the $m$ items to the $n$ agents according to
how the agents value the items. We use the notation
$[n] = \{1,2,\ldots,n\}$ to denote an index set; thus $\N$ indexes the
set of agents and $\M$ indexes the set of items. To formulate our model, we
first treat the items as divisible; we will later explicitly
impose indivisibility (i.e, integrality) requirements. Agents have
preferences over various bundles of items, where a bundle is a subset
of the items. Let $X$ denote the set of all bundles the agents would
be interested in acquiring, and let $\ell = |X|$. In general, this
could be all possible bundles, in which case $\ell = 2^m$, but in some
applications it may be substantially smaller.\footnote{For instance,
  the literature often considers \emph{single-minded} agents. Such
  agents only derive positive value if they acquire a specific,
  designated bundle, and the marginal value of all other items is
  zero. Under single-minded agents we have $\ell \leq n$.}

\paragraph{Notation}
We use the following convention to index vectors and matrices
throughout the paper. For a vector $a \in \bR^\ell$ indexed by the
finite set of bundles, we use the functional notation $a(x)$ for the
component corresponding to $x \in X$. For a vector $a \in \bR^n$
indexed by agents, we use the usual subscript notation $a_i$ for
$i \in [n]$. For a vector $a \in \bR^{n \ell}$ indexed by both, we
write $a_i(x)$ to refer to a component. The convention extends to
matrices. For a matrix $A \in \bR^{n \ell \times d}$, we write
$A_i(x)$ for the $d$-dimensional row associated with agent $i$ and
bundle $x$.
\medskip

\noindent
The primitives of our model will be cast as elements of vector spaces.
A bundle assigned to an agent $i$ is represented as a vector from the
agent's \emph{consumption set}
\begin{equation} \label{def:consumption-set}
\cset_i = \conv\left\{
q_i \in \{0,1\}^\ell : \sum_{x \in X} q_i(x) \leq 1
\right\},
\end{equation}
where `$\conv$' denotes the convex hull operator. The consumption set
is a polytope whose extreme points are in one-to-one correspondence
with bundles; an extreme point $q_i$ is a binary vector corresponding
to the unique $x \in X$ for which $q_i(x) = 1$, or the empty bundle if
$q_i$ is the origin. Fractional vectors represent bundles with
fractional quantities of items. The agents' consumption sets are
identical, but this is not important for our results. We use
$\cset = \cset_1 \times \dots \times \cset_n$ to denote the agents'
joint consumption set.

An \emph{allocation} is represented by a vector $q = (q_1,\ldots,q_n) \in \bR^{n\ell}$, where the subvector $q_i \in \bR^\ell$  encodes the bundle that agent $i$ receives. An allocation is feasible if no more than one unit of each item is supplied to the agents in total. More formally, the set of feasible allocations is captured by the seller's \emph{production set}%
\begin{equation} \label{def:production-set}
F = \conv\left\{
q \in \{0,1\}^{n \ell} :
\sum_{i \in \N} \sum_{x \in X,\, x \ni j} q_i(x) \leq 1 \:\: (j \in \M),\:
\sum_{x \in X} q_i(x) \leq 1 \:\: (i \in \N)
\right\}.
\end{equation}
The production set is a polytope whose extreme points are in
one-to-one correspondence with lists of bundles $(x_1,\ldots,x_n)$
such that each item appears in at most one bundle and each agent
receives no more than one bundle. They therefore correspond to
feasible allocations of indivisible items. Note that, by
definition, if $q \in F$ then $q_i \in H_i$ for each agent $i$, so
that $F \subset H$.

Each agent $i$ has a \emph{valuation} $v_i \in \bR^\ell$ which records
the agent's willingness to pay for each bundle in a common unit of
currency. We assume that each agent's value for the empty bundle is 0.
A \emph{valuation profile} is a vector of agent valuations
$v = (v_1,\ldots,v_n) \in \bR^{n \ell}$.
Given a valuation profile $v$, an allocation $\optq$ is \emph{efficient} (in the economic sense) if it maximizes the total value to the agents among all feasible allocations:
\begin{equation} \label{eq:eff-alloc}
\optq \in \argmax_{q \in \pset} \:\: v^\trans q.
\end{equation}
%

\noindent
Like valuation profiles, \emph{prices} are described by a vector
$p \in \R^{n\ell}$ where $p_i(x)$ is the charge for bundle $x$ to agent
$i$. As defined the prices may be \emph{personalized}, in the sense
that two different agents may see different prices for identical
bundles. We assume that agents have \emph{quasi-linear utility}: given
prices $p$, agent $i$'s utility for bundle $x \in X$ is
$v_i(x) - p_i(x)$. Equivalenty, if $q_i \in H_i$ is the vector
corresponding to bundle $x$, the utility is
$v_i^\trans q_i - p_i^\trans q_i$.  If the auction charges prices $p$
and allocates according to $q \in F$ then the \emph{revenue} totals
$p^\trans q$. The seller has zero value for the items and only derives
utility from the revenue collected.


\medskip
\emph{Price representation. }
We have so far introduced prices as vectors from $\bR^{n\ell}$ which
explicitly list the price of each bundle to each agent. The purpose of
this paper is to analyze the impact of different pricing schemes
(e.g., linear or bundle) on the operation of the auction, especially
its rate of convergence. To impose further structure on prices and
restrict them to a lower-dimensional subspace, we use an indirect
approach that first defines an alternative vector space representation
for the bundles. This approach, explored in depth in
\cite{lahaie2009kernel,lahaie2011kernel}, draws very much from the
class of \emph{kernel methods} used in machine learning. In the
context of prediction and estimation, the intuition behind kernel
methods is that they map the data to a (potentially high
dimensional) space where the function to be learned is
\emph{linear} in the space. The same idea applies in the context of
our auction design, where our goal is to produce a pricing function
that is linear under a particular representation of the bundles.

We introduce a \emph{representation matrix}
$G \in \bR^{n \ell \times d}$, where $m \leq d \leq n \ell$. The row
$G_i(x)$ provides a $d$-dimensional encoding of bundle $x \in X$,
which can also depend on the agent $i$ in general. The interpretation
of the encodings is that they define the ``features'' of the bundles
that are priced in the auction. We take prices to be linear functions
of bundle encodings, so that prices can be represented in $\bR^d$.
That is, the price vector $p \in \bR^{n\ell}$ can be written as
$p = G w$ for some \emph{price parameter} vector $w \in \bR^d$; we
will sometimes drop $p$ and refer to the prices as $G w$. As a
convention, the empty bundle is always encoded as the origin in
$\bR^d$, which means that its price is normalized to 0.
To make these ideas more concrete, let us consider several examples of
representations, each leading to different pricing schemes.
\medskip
%
\begin{description}
\item[Linear] There is a feature for each item ($d = m$). A bundle is
  encoded using its standard 0-1 indicator vector representation. For
  example, with items $a,b,c$, the representations of bundles $\{a,b\}$ and $\{a,b,c\}$ are respectively
$$
\arraycolsep=5pt
\begin{array}{ccccccc}
  &&& a & b & c & \\
  \{a,b\}&\mapsto& [ & 1 & 1 & 0 & ], \\
  \{a,b,c\}&\mapsto& [ & 1 & 1 & 1 & ].
\end{array}
$$
This leads to simple item pricing: the price of a bundle is the sum of the prices associated with each of its items.
In the sequel, the notation $\spoly{1}$ will refer to linear pricing matrices (polynomials of degree 1).

\medskip
\item[Bundle] At the other extreme, we can have a feature for each
  relevant non-empty bundle ($d = \ell$). A non-empty bundle is
  encoded by the unit vector which has a 1 in the component
  corresponding to the bundle. For example, with items $a,b,c$,
  and assuming $X$ consists of all bundles,
  the representations of bundles $\{a,b\}$ and $\{a,b,c\}$ are respectively
$$
\arraycolsep=5pt
\begin{array}{ccccccccccc}
  && & a & b & c & ab & ac & bc & abc & \\
  \{a,b\}&\mapsto & [ & 0 & 0 & 0 & 1 & 0 & 0 & 0 & ], \\
  \{a,b,c\}&\mapsto& [ & 0 & 0 & 0 & 0 & 0 & 0 & 1 & ].
\end{array}
$$
Here each bundle is explicitly priced and there may be no relationship
whatsoever between the prices of different bundles. In the sequel,
$G_{\id}$ will refer to bundle pricing matrices.

\medskip
\item[Polynomial] As a generalization of the linear representation, we
  can have a feature for subsets of items up to size $r$, for $1 \leq r
  \leq m$. A component is 1 if the bundle contains all the items in
  the associated subset, and 0 otherwise.
  For example, with items $a,b,c$, the representations of bundles $\{a,b\}$ and $\{a,b,c\}$ using $r = 2$ are respectively
$$
\arraycolsep=5pt
\begin{array}{cccccccccc}
  & &  & a & b & c & ab & ac & bc \\
  \{a,b\}&\mapsto& [ & 1 & 1 & 0 & 1 & 0 & 0 & ], \\
  \{a,b,c\}&\mapsto& [ & 1 & 1 & 1 & 1 & 1 & 1 & ].
\end{array}
$$
Prices therefore take the form of multi-variate
polynomials\footnote{To see this, note that the price function in this
  example can be written as the quadratic polynomial\\$w_a x_a + w_b
  x_b + w_c x_c + w_{ab} x_a x_b + w_{ac} x_a x_c + w_{bc} x_b
  x_c$, where
  $x_a,x_b,x_c$ are binary indicator variables for the items contained
  in the bundle, and $w$ is the price parameter vector.}
of degree $r$. Note that price components can be positive or negative, so polynomial prices of degree $r \geq 2$ may be super- or sub-additive on various bundles.
In the sequel, $\spoly{r}$ will denote polynomial pricing matrices of degree $r$.
\end{description}

\medskip
\noindent
Note that while the bundle representation is fully expressive, it is
conceptually distinct from the linear and polynomial representations
and should not be construed as a generalization of either.
We are not aware of any iterative auctions in the literature using
price structures other than the ones listed above. However, the
formalism can also accommodate intricate pricing schemes such as
attribute pricing (e.g., square footage for real estate, or population
density for wireless spectrum). The encodings cannot be entirely
arbitrary: they must retain enough information about the items
contained in the bundles to allow one to verify, given a vector of $n$
bundle encodings, whether it forms a feasible allocation. The above
encodings all have this property---see~\cite{lahaie2011kernel} for a
more detailed discussion.

The representations given above (linear, bundle, and polynomial) are
all anonymous, in the sense that $G_i(x) = G_j(x)$ for all $i,j \in
\N$, leading to anonymous prices. To personalize prices, we can
include features that depend on the identity of the agent. For
instance, one can introduce an agent intercept to any of the encodings
above, leading to a feature space of dimension $d + n$. To obtain
entirely separate prices for each agent, one can take $n$ copies of
the original feature space and map a bundle into the agent's copy; the
resulting dimension is $nd$. Our results can be adapted to such
schemes by extending the dimension of the feature space accordingly.


\section{Allocation and pricing.}
\label{sec:opt}

It has been understood since the literature on stability of equilibria that price adjustment processes minimize a convex potential involving indirect utilities~\cite{arrow1971general,varian1981dynamical}. In the context of combinatorial auctions, this potential arises as the dual of the efficient allocation problem~\citep{bikhchandani2002linear,bikhchandani2002package}. We give a full development of this duality for our setup. This will serve to clarify why iterative auctions can be analyzed as subgradient methods, and also provides bounds on price magnitudes, important for later analysis.

Consider the agents' \emph{aggregate} indirect utility function and demand correspondence, defined respectively as
\begin{eqnarray}
u(p;v) & = & \max_{q \in \cset}\: v^\trans q - p^\trans q, \label{eq:agg-demand-util} \\
U(p;v) & = & \argmax_{q \in \cset}\: v^\trans q - p^\trans q, \label{eq:agg-demand-corresp}
\end{eqnarray}
for personalized bundle prices $p \in \bR^{n\ell}$. We will often
suppress parameter $v$ when clear from context. As $\cset$ is a
product set, an optimal $q$
in~(\ref{eq:agg-demand-util}--\ref{eq:agg-demand-corresp}) decomposes
into $q = (q_1,\ldots,q_n)$ where $q_i$ maximizes agent $i$'s utility
$v_i^\trans q_i - p_i^\trans q_i$ individually over $H_i$. The indirect
utility $u$ therefore aggregates the agents' maximal utilities over
bundles, given prices $p$. Similarly, correspondence $U$ maps to
vectors of utility-maximizing bundles, listing one bundle for each
agent, which may overlap in general.
We also have the seller's indirect utility function and supply
correspondence, defined respectively as
\begin{eqnarray}
s(p) & = & \max_{q \in F}\; p^\trans q, \label{eq:supply-util} \\
S(p) & = & \argmax_{q \in F}\; p^\trans q. \label{eq:supply-corresp}
\end{eqnarray}
We say that prices $\optp$ are \emph{market clearing} if
$U(\optp) \cap S(\optp) \ne \emptyset$, and for any allocation
$\optq \in U(\optp) \cap S(\optp)$ we say that prices $\optp$
\emph{support} $\optq$. More explicitly, this means that: 1)
$\optq \in \pset$, so it represents a feasible allocation; 2) the
assigned bundles according to $\optq$ maximize the agents' utilities
at prices $\optp$; 3) allocation $\optq$ maximizes the seller's
revenue at prices $\optp$. Since the agents and seller would each
willingly select allocation $\optq$ when faced with prices $\optp$,
the prices are market clearing in this sense.
If there exists prices $\optp$ supporting an allocation $\optq$, then
the allocation is efficient, because for any $q\in F$ we have
\begin{equation}\label{eq:eff-alloc-derivation}
v^\top \optq = (v^\top \optq - \optp^\top \optq) + \optp^\top \optq \geq (v^\top q -
\optp^\top q) + \optp^\top q = v^\top q,
\end{equation}
where the inequality holds because $\optq \in U(\optp)$
and $\optq \in S(\optp)$. An iterative auction proceeds by updating a
provisional allocation and prices given agent bids in each round,
until the prices support the allocation. According to
derivation~\eqref{eq:eff-alloc-derivation}, this provides a
certificate that an efficient allocation has been reached.

The convex dual to the problem of computing an efficient allocation is
the problem of finding market clearing prices. As the minimization
problem is on the price parameter, we will refer to the price
optimization as the \emph{primal} objective and the allocation
optimization as the \emph{dual} objective.\footnote{This is consistent
  with the usual convention in convex analysis, but unfortunately
  conflicts with the convention in the auctions literature, where the
  primal is typically the allocation problem and the dual is the
  pricing problem. We adopt the former convention because it becomes
  much simpler to directly apply convex analysis results in the proofs.}
In what follows, `$\ker$' refers to the kernel (i.e., nullspace) of a matrix.
Also, given any two subsets of Euclidean space $A,B$, the sum $A + B$ is defined in the natural way: $A + B = \{ x + y : x \in A, y \in B \}$.
\begin{theorem} Consider the optimization problems
\label{thm:alloc-price-dual}
\begin{align}
  &\inf \left\{u(Gw) + s(Gw) : w \in \bR^d \right\},\label{eq:alloc-price-primal} \\
  &\sup\left\{v^\top q : q \in \cset \cap (\ker(G^\top) + F) \right\}. \label{eq:alloc-price-dual}
\end{align}
Then the primal value in~\eqref{eq:alloc-price-primal} equals the dual value in~\eqref{eq:alloc-price-dual},
and both primal and dual optima are attained.
Moreover, allocation $\optw$ and price parameter $\optq$ are optimal primal and dual solutions if 
\begin{equation}\label{eq:comp-slack}
  \optq \in U(G\optw) \cap \Big( \ker(G^\top) + S(G\optw) \Big).
\end{equation}
%
\end{theorem}
To understand the dual formulation~\eqref{eq:alloc-price-dual}, which
captures the allocation problem, note first that the objective is the
efficiency $v^\top q$ of allocation $q$. The first part of the
constraint is simply $q \in \cset$, meaning agents are allocated bundles
from their consumption sets. This is combined with the constraint
$q \in \ker(G^\top) + F$. If the latter set were simply $F$, then the
feasible set would reduce to $\cset \cap F = F$, which is simply the
convex hull of set of feasible allocations. Recall that the role of
$G$ is to constrain the possibilities for the pricing space. Dually,
restricting the dimension of the rows in $G$ expands $\ker(G^\top)$,
and $\ker(G^\top) + F$ becomes a relaxation of $F$. This discussion
leads to the following result.
\begin{corollary}
\label[corollary]{integrality}
 If $G$ has full row rank, then there is an integer optimal solution
 $\optq$ to the dual objective~\eqref{eq:alloc-price-dual}, and for any
 primal optimal solution $\optw$, prices $\optp = G\optw$ support
 allocation $\optq$.
\end{corollary}
For instance, $G$ has full row rank if it encodes personalized bundle
prices, or personalized polynomial prices of degree $m$, which shows
that such prices can support efficient integer allocations (i.e.,
allocations of indivisible items). The former case was first proved
by~\citet{bikhchandani2002package}. If the set of relevant bundles $X$
is restricted, lower-dimensional prices may suffice. The condition
given in~\Cref{integrality} is sufficient but not necessary, so in
practice lower-dimensional prices may still be able to clear the
market even when $X$ is large. In order to interpret the results in
the sequel, one can assume that $G$ is sufficiently expressive to
ensure that~\eqref{eq:alloc-price-dual} has an integer optimal
solution; if this is not the case, our results still meaningfully
bound convergence rates under divisible items.

Let us now consider the primal more closely. We write $\scrD : \R^d \to \R$ to represent the primal objective function over price parameters $w$, parametrized by the valuation vector $v$:
\begin{equation*}
  \scrD(w;v) = u(Gw;v) + s(Gw).
\end{equation*}
Recall that the prices are $p = Gw$. The first term is aggregate
indirect utility, namely the maximum surplus that agents can achieve
by each selecting their preferred bundle at prices $p$ (again, the
bundles may overlap). Prices should be set high to minimize this term.
The second term is the seller's indirect utility, namely the maximum
revenue possible over all feasible allocations, under prices $p$.
Prices should be set low to minimize this term. Thus the two terms
lead prices to strike a balance between demand and supply.
Looking ahead to the iterative auction of
Section~\ref{sec:iter-auction}, the auction can be construed as a
subgradient method to optimize the primal
objective~\eqref{eq:alloc-price-primal}, obtaining subgradient
information from the agents' bids, which lie in the dual space
associated with objective~\eqref{eq:alloc-price-dual}.

A difficulty with formulations~\eqref{eq:alloc-price-primal}
and~\eqref{eq:alloc-price-dual} is that the optimal prices in the
primal may not be unique, as both primal and dual are linear programs.
As a result it is not possible analyze convergence of the actual
prices, only convergence in objective value. To obtain a strictly
convex objective and unique solution, we can introduce a
regularization term $\norm{w}_2^2/2 $ with weight $\lambda > 0$ into
the primal objective. The duality result generalizes as follows.
\begin{theorem} \label{thm:alloc-price-dual-reg}
Given weight $\lambda > 0$, consider the optimization problems
\begin{align}
  \label{eq:alloc-price-primal-reg}
  & \inf \left\{u(Gw) + s(Gw) + \frac{\lambda}{2} \|w\|_2^2 : w \in \bR^d \right\}, \\
  \label{eq:alloc-price-dual-reg}
  & \sup\left\{v^\top q - \frac{1}{2\lambda}\| G^\top (q - q') \|_2^2 : q \in \cset, q' \in F \right\}.
\end{align}
Then the primal value in \eqref{eq:alloc-price-primal-reg} equals the dual value in \eqref{eq:alloc-price-dual-reg}, and both the primal and dual optima are attained. Furthermore, the primal optimum $\optw$
is unique. Price parameter $\optw$ and allocations $(\optq, \optq')$ are optimal primal and dual solutions if
\begin{equation}\label{eq:comp-slack-reg}
\optq \in U(G\optw),\qquad \optq' \in S(G\optw), \qquad G^\top (\optq -
\optq') = \lambda \optw.
\end{equation}
\end{theorem}
The original duality relation derived in Theorem~\ref{thm:alloc-price-dual} can be viewed as the limiting case of Theorem~\ref{thm:alloc-price-dual-reg} when we set $\lambda = 0$. To make this more transparent, first observe that the constraint in the dual objective~\eqref{eq:alloc-price-dual} can be re-written as
\begin{eqnarray}
G^\top q = G^\top q' &\ & (q \in \cset,\, q' \in F).
\end{eqnarray}
To relax the constraint in this form, we can replace it with a penalty term in the objective that quantifies the discrepancy between demand and supply according to the squared norm, weighed according to $\lambda > 0$:
$$
\frac{1}{2\lambda} \|G^\top (q - q')\|_2^2 .
$$
This leads to the relaxed dual formulation
in~\eqref{eq:alloc-price-dual-reg}. We will write $\scrD_{\lambda}$ to
refer to the objective~\eqref{eq:alloc-price-primal-reg} with
regularization term weighed by $\lambda > 0$. As $\lambda$ tends to
$0$, the original constraint is satisfied exactly (supply matches
demand), so by convention $\scrD_0 \equiv \scrD$. In the regularized
primal formulation \eqref{eq:alloc-price-primal-reg}, $\lambda > 0$
shrinks the price parameter vector towards 0.
The proof of Theorems~\ref{thm:alloc-price-dual}
and~\ref{thm:alloc-price-dual-reg} appears in
Appendix~\ref{sec:proofs}. The case of $\lambda = 0$ follows from
linear programming duality. When $\lambda > 0$, general convex duality
(Fenchel duality) is invoked. The duality proof admits other choices
for the convex regularizer besides the squared norm, which could lead
to improved bounds on convergence rates.

The following result bounds the magnitude of the optimal price parameter, which is an essential element of our convergence rate analyses.
\begin{proposition} \label[proposition]{prop:w-bounds}
  Consider the optimization problems in~\eqref{eq:alloc-price-primal} and~\eqref{eq:alloc-price-primal-reg},
  and set $V = \|v\|_\infty$.
  \begin{itemize}
    \item
      If $G = G_{\id}$, then every optimal $\bar w$ satisfies
      $\|\bar w\|_\infty \leq (n+1)V$
      and
      $\|\bar w\|_2 \leq (n+1) V \sqrt{\ell}$.
   \item
      If $G = \spoly{r}$ for some integer $r\geq 1$, then every optimal $\bar w$ satisfies
      $\|\bar w\|_\infty \leq (n+1) V 2^r$
      and
      $\|\bar w\|_2 \leq (n+1) V 2^r m^{r/2}$.
  \end{itemize}
\end{proposition}
The $V$ bound here amounts to a choice of units. It could be normalized to $V = 1$, but we choose to keep it explicit to clarify how it affects the choice of step-size in the auction.
%
%
The proposition establishes that the optimal price parameter (even when non-unique) lies in a bounded set. Importantly, the size of this set depends on the number of degrees of freedom in $w \in \R^d$, meaning the dimension $d$. This implies that the price parameters $w$ can be constrained during the auction to lie in a ball of sufficiently bounded radius without affecting the optimum. We will see that the radius of this ball affects the convergence rate of the auction; simply put, increasing the representation power of $G$ increases the radius of the ball.

As a sketch, the proof first proceeds by bounding $\norm{p}_{\infty}$,
where recall that $p = Gw$ lists the explicit bundle prices. The
reasoning is that as the components of $p$ become large, the second
term of $\scrD$ grows, and as the components of $p$ become small, the
first term grows. Taken together, it follows that the optimum must be
bounded.
We then obtain bounds on $\norm{w}_{\infty}$ given the bound on
$\norm{p}_{\infty}$. For $G = \Gid$ this is trivial because the latter
is a stack of identity matrices. The case of $G = \Gpoly{r}$ requires
a combinatorial argument. Finally, a bound on $\norm{w}_{\infty}$
immediately yields a bound on $\norm{w}_2$.


\section{Iterative auction.}
\label{sec:iter-auction}

\begin{figure}
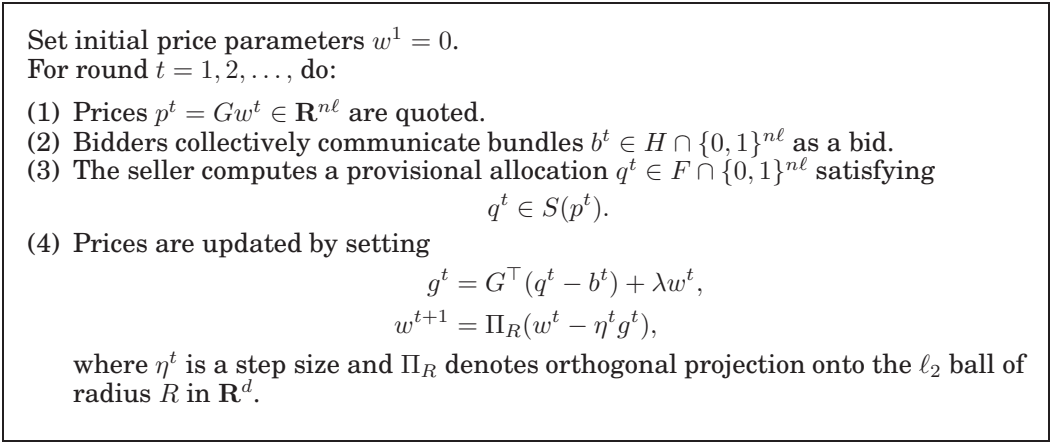

\begin{framed}
Set initial price parameters $w^1 = 0$.\\
For round $t = 1, 2, \ldots,$ do:
    \begin{enumerate}
      \item \label{step:price-quote}
        Prices $p^t = G w^t \in \R^{n\ell}$ are quoted.
      \item \label{step:agents-bids}
        Bidders collectively communicate bundles $b^t \in \cset \cap \{0,1\}^{n\ell}$ as a bid.
      \item \label{step:prov-allocation}
        The seller computes a provisional allocation
        $q^t\in F \cap \{0,1\}^{n\ell}$ satisfying $$q^t \in S(p^t).$$
     \item \label{step:price-update} Prices are updated by setting
        \begin{align*}
          g^t
          &=
          G^\top(q^t -  b^t) + \lambda w^t,
          \\
          w^{t+1}
          &=
          \Pi_R(w^t - \eta^{t} g^t),
      \end{align*}
      where $\eta^t$ is a step size
      and $\Pi_R$ denotes orthogonal projection onto the
      $\ell_2$ ball of radius $R$ in $\R^d$.
    \end{enumerate}
\end{framed}
\caption{Subgradient auction.  \label{fig:subg-auction}}
\end{figure}



We now introduce the iterative auction that will be the subject of our
analysis.  The auction is described in Figure~\ref{fig:subg-auction}
and follows a high-level outline common to many practical auction
designs: 1) prices are quoted; 2) agents place bids on bundles,
meaning offers to purchase bundles at the quoted prices; 3) the seller
computes a provisional allocation based on bids and prices; 4) prices
are updated when there is discrepancy between the bids and
allocation. Note that the agents do not make price offers, so this is
a \emph{clock auction}.

The auction is fully specified given the matrix $G$, the
regularization parameter $\lambda \geq 0$, the step sizes $\eta^t$,
and the projection radius $R$. Since different choices for the matrix
$G$ correspond to different pricing schemes,
Figure~\ref{fig:subg-auction} may more generally be viewed as a class
of auctions. The price update in step~\ref{step:price-update} has a
natural economic interpretation. Bid $b^t$ represents demand, while
allocation $q^t$ represents supply, so $q^t - b^t$ corresponds to
\emph{excess supply} in bundle space. By applying $G^\top$, this maps
to excess supply in feature space, where the price parameter vector
lies. The result is then subtracted away from the current price
parameter $w^t$, which is therefore updated in the direction of
\emph{excess demand}. With regularization ($\lambda > 0$), the update
is pulled back towards the current iterate $w^t$, leading to more
conservative price parameter updates.

The projection in step~\ref{step:price-update} is essentially a
rescaling to ensure that price coordinates are well-behaved, and we
confirm below that the radius can be chosen large enough (relative to
the number of agents and items) to leave the final clearing prices
unaffected. Projection may seem undesirable, since small choices of
$R$ can rule out the (unconstrained) optima guaranteed to exist for
the pricing problem
in~\Cref{thm:alloc-price-dual,thm:alloc-price-dual-reg}. However, as
detailed in \Cref{prop:w-bounds}, it is possible to choose $R$ purely
from knowledge of the matrix $G$ and units $V$ so that the ball of
radius $R$ includes unconstrained optima. This point will rearise in
\Cref{sec:stochastic,sec:arbitrary}, where the iterative auction is
shown to converge to these unconstrained optima under an appropriately
chosen $R$.

Note that in step~\ref{step:agents-bids} of the auction, agents are
required to provide an \emph{integer} bid vector, representing a
bundle with whole items. This is where item indivisibility comes into
play. Under our model, requiring integer bid vectors is unrestrictive,
because an agent's best-response problem is to optimize a linear
function (valuation minus price) over the consumption
set~\eqref{def:consumption-set}, whose extreme points are
integer. Therefore, there is always an integer best-response---in
plainer terms, an agent can always maximize its utility by bidding on
a bundle with whole items. Similarly, the seller is required to
compute an integer allocation in step~\ref{step:prov-allocation}. As
prices are a linear function and the production
set~\eqref{def:production-set} has integer extreme points, this does
not prevent the seller from computing a revenue-maximizing allocation.

It is worth noting that in our auction neither the price parameters
$(w^t)_{t\geq 1}$ nor the bundle prices $(p^t)_{t\geq 1}$, where $p^t
= Gw^t$, are constrained to move monotonically (ascending or
descending pointwise). While most iterative auctions in the literature
are monotone, there is precedent for non-monotonic price paths, for
instance Ausubel's auction for multiple distinct
items~\citep{ausubel2006efficient}. On one hand, non-monotone price
paths allow for highly adversarial bidding behavior that can even
prevent the auction from converging; we return to this issue
in~\Cref{sec:arbitrary} and show how it is addressed by activity
rules. On the other hand, non-monotone auctions also have substantial
advantages when it comes to convergence rates, which makes them more
natural to study in the context of this paper.\footnote{To illustrate,
consider the task of guessing a secret number: One player chooses a
number $y$ within the interval $[0,1]$, and another player must guess
this number to within accuracy $\epsilon > 0$. The game protocol is
that the guesser specifies an $x \in [0,1]$, and the chooser responds
with $\1[ x \geq r]$. With monotonic guesses, it is necessary to make
$\Omega(1/\epsilon)$ queries; on the other hand, non-monotonic
questions allow binary search, meaning $\cO(\ln(1/\epsilon))$
queries.}

We next confirm that there is a direct relationship between the
optimization problems of Section~\ref{sec:opt} and the iterative
auction in Figure~\ref{fig:subg-auction}. First, observe that the
auction makes no mention of agent valuations---valuations affect agent
bidding behavior, but cannot be part of the auction specification. We
say that a bid vector $b \in \cset$ and valuation vector $v \in
\bR^{n\ell}$ are \emph{consistent} with each other at prices $p$ if $b
\in U(p;v)$. In words, this means that the bid vector is a
best-response to prices assuming the given valuations (i.e., the bid
maximizes the agents' aggregate utility).

\begin{proposition}
  \label[proposition]{fact:subgradient}
  Let representation matrix $G$ and regularization parameter $\lambda \geq 0$ be given.
      If bid vector $b^t \in \cset$ is consistent with valuation $v^t$ at prices
      $p^t = Gw^t$, and allocation $q^t \in F$ maximizes the seller's
      revenue at prices $p^t$, then
      \begin{equation}
      g^t = G^\top (q^t - b^t) + \lambda w^t
      \label{eq:subgradient}
      \end{equation}
      is a subgradient of $\scrD_{\lambda}(\,\cdot\, ; v^t)$ at $w^t$.
\end{proposition}
%
\noindent
The full proof is given in the appendix.
This proposition connects the optimization problem and the iterative
auction; specifically, since the subgradient expression
in~\eqref{eq:subgradient} matches the update step in the auction,
namely step~\ref{step:price-update} in Figure~\ref{fig:subg-auction},
the auction is performing subgradient descent on $\scrD_{\lambda}$. Under this
viewpoint we can analyze auction convergence using well-developed
techniques.

%
\paragraph{Implementation}
Athough our focus is on the theoretical convergence properties of the
subgradient auction in Figure~\ref{fig:subg-auction}, let us say a few
words about how it might be implemented in practice. This is a salient
question given that the auction, as specified, manipulates
high-dimensional vectors and computes allocations over a complex polytope.

Note first that although bids $b^t$ and allocations $q^t$ have
dimension $n\ell$, and $\ell$ may be exponential in $m$, the
restriction that they should be integer in
steps~\ref{step:agents-bids} and~\ref{step:prov-allocation} means that
they have at most $n$ non-zero entries and are succinct to
communicate. The sparsity of bids and allocations also means that the
price update in step~\ref{step:price-update} can be computed
efficiently. As mentioned previously, the projection in the price
update is simply a rescaling of the price parameters.

As the representation matrix $G$ is commonly known, communicating
prices is a matter of communicating $w^t \in \bR^d$, which can be done
directly for linear prices or polynomials of low degree. For larger
$d$, the following result gives a dual representation of $w^t$ with
storage on the order of $nt$ coefficients. In the statistical learning
literature, such results are known as \emph{representer
  theorems}~\citep{SVM}.

\begin{proposition}
  \label[proposition]{fact:subgradient:representer}
  Let matrix $G$ and regularization parameter $\lambda \geq 0$ be given.
  For all $t$, $w^t \in \SPAN(G^\top)$.
  Indeed, setting $\gamma^1 = 1$
  and $\gamma^{t+1} = R / \max\{R, \norm{w^t - \eta^t g^t}_2  \} $,
  $w^t$ has the form
  \[
    w^t
    = G^\top \sum_{s=1}^{t-1} \gamma^{t} \eta^s (b^s - q^s) \prod_{j =
      s+1}^{t-1} \gamma^{j+1}(1  - \lambda \eta^j),
  \]
  meaning $w^t$ may be reconstructed from $(\eta^s, b^s, q^s, \gamma^s)_{s=1}^{t-1}$.
\end{proposition}
\noindent
The proof of this statement can be found in the appendix.
The claim follows by an inductive argument, since by the price update
rule $w^{t}$ is a linear combination of elements from $\SPAN(G^\top)$.

The remaining item is to compute a revenue-maximizing allocation in
step~\ref{step:prov-allocation}. The efficient allocation problem
in~\eqref{eq:eff-alloc} is known to be NP-hard by reduction from
weighted set packing~\citep{nisan2000bidding}, and
comparing~\eqref{eq:supply-util} and~\eqref{eq:eff-alloc}, we see that
the allocation step is equivalently intractable. In practice, this
step is implemented using integer programming solvers, and there is a
large body of work on solvers for combinatorial
auctions~\citep[e.g.][Part III]{cramton2006combinatorial}.
Using an efficient allocation oracle in the specification allows us to focus
on the number of auction rounds, which is the relevant metric for
iterative combinatorial auctions.

There are several parameters to tune to run the auction: the
regularization weight $\lambda$, the projection radius $R$, and the
step size schedule $\{\eta^t\}$. Our convergence rate analyses---see
Theorems~\ref{fact:sgd:basic:1} and~\ref{fact:sgd:basic:2}---will
provide guidance on how to set each of these parameters. Essentially,
all of them can be set in terms of $V$, the maximum possible agent
value for a bundle.

The convergence results we present in the next sections concern price
convergence (in the primal), not allocation convergence (in the dual),
and even for prices convergence holds only in the limit rather than in
a finite number of steps. For auctions based on the subgradient
method, it is not possible to establish finite-time convergence
without further structural assumptions~\citep{de2007ascending}. To
increase the chances of matching supply with demand at each round,
practical auction designs break ties in
step~\ref{step:prov-allocation} to satisfy as many agents as possible,
and also allow agents to take an $\epsilon$-discount on bundles in the
provisional allocation~\citep[e.g.][]{parkes1999bundle}. For
relatively small $\epsilon$, this does not impact our bounds on
convergence rate.



\section{Stochastic bidding.}
\label{sec:stochastic}

In this section and the next, we provide the central results of the
paper which bound the convergence rate of prices in the subgradient
combinatorial auction of Figure~\ref{fig:subg-auction}. We first
consider \emph{stochastic bidding}, meaning that agents behave
according to a random utility model. In the next section we turn our
attention to \emph{adversarial bidding} in which agents can place
arbitrary bids, constrained across rounds only by activity rules.

Our stochastic model aims to capture the fact that bids can
incorporate an element of randomness at each round due to fluctuating
valuations, bounded rationality, behavioral noise, etc. However,
rather than directly assume that bids are stochastic, we instead
assume that valuations are stochastic at each round and that bids are
chosen as best-responses to prices according to the realized
valuations. This is the bidding behavior that arises from a random
utility model, familiar from discrete choice
modeling~\citep{mcfadden1973conditional}.
Formally, at each round the agents draw their valuation profile $v^t$
from a fixed distribution, denoted $\nu$. That is, one should view
each $v^t$, for $t = 1, 2, \ldots,$ as an i.i.d.\ draw from $\nu$.
Once $v^t$ is drawn and the prices $p^t = Gw^t$ are quoted by the
auction, the agents place a collective bid vector $b^t$ consistent
with $v^t$, where ties are broken arbitrarily in case the
best-response is not unique.
In discrete choice models, the random valuation is usually decomposed
as $v^t = \tilde{v} + \epsilon^t$, where $\tilde{v}$ is the mean
valuation and $\epsilon^t$ is an error term capturing deviations from
the mean at round $t$. The most common error models for $\epsilon^t$
are the Gumbel distribution (known as the logit model) and the
Gaussian distribution (known as the probit model).

Under stochastic bidding, bid vector $b^t$ is consistent with random
valuation $v^t$ under prices $p^t$ at each round, which means that the
auction is performing subgradient descent as per
Proposition~\ref{fact:subgradient}. Note that the associated
distribution $\nu$ may be arbitrary, and in particular $\nu$ does not
need to be a product distribution across the $n$ bidders. Our
convergence results are robust to arbitrary correlations between the
bidders valuations. However, one limitation of the model is that the
i.i.d.\ nature of the distribution cannot incorporate learning from
past bids and prices (e.g., as one would expect if there were a common
value component to the agents' valuations).

\begin{theorem}
  \label{fact:sgd:basic:1}
  Let $\nu$ be a distribution over value vectors $v\in \R^{n\ell}$
  with $V = \Exp_\nu[ \|v\|_\infty ] < \infty$;
  by these conditions, there exists an optimum $\bar w$
  to the problem $\min\{\Exp_\nu[ \scrD_{\lambda}(w;v)] : w\in \R^d\}$.
  Moreover, with probability at least $1 - \delta$ over an i.i.d. draw of valuations $(v^t)_{t=1}^T$,
  running an iterative auction over $T$ rounds
  with step size $\eta^t = V/\sqrt{t}$,
  regularization $\lambda \leq 1/V$,
  and any projection radius $R \geq \|\bar w\|_2$
  gives the bound
  \begin{align*}
    \Exp_\nu[\scrD_{\lambda}(\hat w^T; v)] - \Exp_\nu[\scrD_{\lambda}(\bar w; v)]
    &\leq
    \sum_{t=1}^T \hat\eta^t\left(
      \Exp_\nu[\scrD_{\lambda}(w^t; v)] - \Exp_\nu[\scrD_{\lambda}(\bar w; v)]
    \right)
    \\
    & \leq
        \cO \left( \frac{\kappa^2 \ln T \sqrt{\ln(1/\delta)}\, V}{\sqrt{T}} \right)
  \end{align*}
  where $\hat\eta^t = \eta^t / \sum_{s=1}^T \eta^s$ and
  $\hat w^T = \sum_{t=1}^T \hat\eta^t w^t$ is the averaged iterate.
  The quantity $\kappa$ depends on representation matrix $G$ and
  projection radius $R$, and may be bounded as follows.
  \begin{itemize}
    \item
      When $G = G_{\id}$, it suffices to choose $R = (n+1)V\sqrt{\ell}$,
      whereby $$\kappa \leq \sqrt{n} + \sqrt{m} + 2(n+1)\sqrt{\ell}.$$
    \item
      When $G = \spoly{r}$, it suffices to choose $R = (n+1)V m^{r/2} 2^r$,
      whereby $$\kappa \leq (1 + \sqrt{n}) m^r + 2 (n+1) m^{r/2} 2^{r}.$$
  \end{itemize}

\end{theorem}
There are several terms in the bound of \Cref{fact:sgd:basic:1}, but
the leading term $\kappa^2$ roughly reflects the number of degrees of freedom
(i.e., the dimension $d$ of the price parameter $w$), and we see that
increasing polynomial degree or using bundle pricing weakens
guarantees on convergence time. The quantity $V$ essentially captures the
scale of the bidder valuations. In the simplest case where $\nu$ has
compact support, $V$ corresponds to the largest possible value for a
bundle. Lemma~\ref{fact:Vbounds} in the appendix provides bounds on $V$ for the logit
and probit models, as well as any error distribution with subgaussian
tails; all bounds have the form
$\norm{\bar{v}}_\infty + \cO(\sigma_{\max} \ln(n\ell))$,
where $\bar v$ is the mean valuation
and $\sigma_{\max}$ is the maximum over the valuation's coordinate-wise standard deviations.


It is worth stressing how the bound in \Cref{fact:sgd:basic:1} (as well as the upcoming bound in \Cref{fact:sgd:basic:2}) departs from standard statistical treatments.  In statistical learning theory, it is standard to choose either the radius $R$, or the regularization weight $\lambda > 0$, so as to provide faster or slower convergence.  The same applies to the non-sequential (batch) setting, where a typical bound for kernel classifiers depends purely on $\lambda > 0$ \citep[Corollary 4.3]{bbl_esaim}.
Such an approach is not possible here because $R$ must
be chosen so as to leave the set of optimal solutions
intact---otherwise, the optimal solution would lose its meaning as
clearing prices. The main challenge in proving~\Cref{fact:sgd:basic:1}
is to show how this can be achieved with bounded choices of $R$.

\Cref{fact:sgd:basic:1} only controls convergence of the objective
function $\scrD_{\lambda}$, not the prices themselves. If we want to claim that
the individual bundle prices during the auction are `informative' to
the bidders, then the coordinates of $p^t = Gw^t$ should also be
stable. To this end, we provide the following result.
\begin{corollary}
  \label[corollary]{fact:sgd:params:1}
  Consider the setting of~\Cref{fact:sgd:basic:1} but with $\lambda \in (0,1/V]$,
  and let $w\in \R^d$ be any vector satisfying
  $\Exp_\nu[\scrD_{\lambda}(w; v)] \leq \Exp_\nu[\scrD_{\lambda}(\bar w; v)] + \epsilon$ for some $\epsilon \geq 0$.
  Then \mbox{$\|w - \bar w\|_2^2 \leq 2\epsilon / \lambda$}.
  In particular, after $T$ rounds, with probability at least
  $1-\delta$ prices $\hat p^T = G\hat w^T$ and $\bar p = G\bar w$ satisfy
  \[
    \norm{\hat p^T - \bar p}_\infty
    \leq
    \cO \left( \frac{\kappa \|G\|_{2,\infty}}{\lambda}
    \sqrt{ \frac
      {\ln(T)\sqrt{\ln(1/\delta)}}{\sqrt T}}\right),
  \]
  where $\kappa$ may be bounded as in~\Cref{fact:sgd:basic:1},
  and $\|G\|_{2,\infty} = \max\{ \|Gw'\|_\infty : \|w'\|_2 \leq 1\}$
  can be bounded as $\|G_{\id}\|_{2,\infty} \leq 1$
  and $\|\spoly{r}\|_{2,\infty} \leq m^{r/2}$.
\end{corollary}
%
\noindent
In words, this statement converts the convergence in objective value from \Cref{fact:sgd:basic:1} to convergence in prices themselves, assuming $\lambda > 0$.
%
%
However, \Cref{fact:sgd:params:1} has a few weaknesses: (1)
$\lambda > 0$ must be chosen small to ensure there is not too much
discrepancy between demand and supply, thus the bound converges
slowly; (2) the economic meaning of regularizing $\bar{w}$ is still
unclear---it appears to favor a bidder-optimal choice of prices, but
we have no formal statements to this effect.
Note that the choice $\lambda = 1/V$ causes the right hand side to
scale linearly with $V$, matching \Cref{fact:sgd:basic:1} and also the interpretation
of $V$ as units or scale.

\paragraph{Bias-Variance}
The results of this section can be interpreted in terms of a
bias-variance trade-off. The infimal value of
$\Exp_\nu[\scrD_{\lambda}(\,\cdot\,; v)]$ depends on the representation matrix
$G$, and in this way represents the `bias' of the auction. Using a
more expressive class of prices reduces this bias, and the lowest bias
is attained when the prices can support an integer optimal solution.
On the other hand, the actual value of the bounds, which gives the
rate of convergence of the auction, is the `variance' term. Increasing
price expressiveness weakens the bound, and therefore simpler pricing
matrices (e.g., polynomial matrices of low degree) should exhibit
faster convergence. One way to moderate the trade-off is to use the
simplest class of prices available (in terms of dimensionality) that
clears the market, although this can be hard to know a priori.


\section{Arbitrary bidding.}
\label{sec:arbitrary}

In this section we turn to a model where agent bids can be essentially
arbitrary, and even adversarial across rounds. While the model is
behaviorally unreasonable without further constraints, it is still
possible to provide a certain convergence guarantee on the
objective value. On the other hand, it is not possible to obtain price
convergence: the optimal price vectors may drift and oscillate, in
contrast with the conclusions we were able to draw in
\Cref{fact:sgd:params:1}. The model therefore motivates the use of
\emph{activity rules} to constrain agent bids across rounds, and this
section shows how a well-designed activity rule can result in
meaningful convergence guarantees for both the objective and prices.


The arbitrary bidding model is as follows.  In round $t$, upon seeing
prices $p^t = Gw^t$, bidders collectively release an integer  bid
vector $b^t \in \cset \cap \{0,1\}^{n\ell}$ as specified in
step~\ref{step:agents-bids} of the auction.
In contrast with the stochastic model of the previous section, where
$b^t$ must be a best-response with respect to the random valuation vector
$v^t$ according to~\eqref{eq:agg-demand-corresp}, $b^t$ need only be
consistent with some valuation vector $v^t$.
Now, without any constraints on the space of valuations, one can
always find a valuation vector with which a given bid vector is
consistent, whatever the prices.
The first result of this section merely requires there to exist choices of
\mbox{$(v^t)_{t \geq 1}$} which satisfy \mbox{$\norm{v^t}_\infty \leq V$}
for some scalar $V$. This is a mild constraint which effectively means
that $b_i^t(x)$ must be zero whenever $p_i^t(x)$ exceeds $V$.

Our initial result under arbitrary bidding is the following.
Superficially, the statement appears similar to the convergence
statement for the stochastic model given in \Cref{fact:sgd:basic:1}.
The essential difference is that the left-hand side is no longer
competing with a fixed target $\inf_w \Exp_\nu[\scrD_{\lambda}(w; v)]$. Instead,
the comparison is against $\inf_w \sum_{t \in [T]} \scrD_{\lambda}(w; v^t)$, and
each term in the summation can change drastically at each round $t$.
This objective does have an economic interpretation if one views our
procedure as a sequential posted price mechanism rather than an
iterative auction. A new set of $n$ bidders arrives at each round, and
a new set of $m$ items is available for sale. The seller's problem is
to try to post prices that clear the market at each round $t$, before
bidder valuations $v^t$ are revealed, where clearing quality is
captured by the objective $\scrD_{\lambda}(\,\cdot\, ; v^t)$. The result bounds
the \emph{regret} of the procedure against the best fixed prices in
hindsight, which is a standard objective for online
algorithms~\cite{hazan2012survey}.
\begin{theorem}
  \label{fact:sgd:basic:2}
  Consider an iterative auction where bid vectors $(b^t)_{t=1}^T$,
  with some consistent sequence of value vectors $(v^t)_{t=1}^T$, are
  announced in alternation with price parameters $(w^t)_{t=1}^T$
  provided by the auction mechanism invoked with step size
  $\eta^t = V/\sqrt t$ for some $V\geq 0$,
  regularization $\lambda \leq 1/V$,
  and some projection radius
  $R\geq 0$. Then there exists a minimizer $\bar w^T$ to
  $f(w) = \sum_{t=1}^T \eta^t \scrD_{\lambda}(w; v^t)$, and if
  $V \geq \sup_{t\in [T]} \|v^t\|_\infty$ and $R \geq \|\bar w^T\|_2$,
  then
%
  \begin{align*}
    \sum_{t=1}^T \hat\eta^t\left(
      \scrD_{\lambda}(w^t; v^t) - \scrD_{\lambda}(\bar w^T; v^t)
    \right)
    &\leq
    \cO \left( \frac {\kappa^2 V \ln T}{\sqrt T}\right),
  \end{align*}
  where $\hat\eta^t = \eta^t / \sum_{s=1}^T \eta^s$.
  The $\kappa$ quantity depends on representation matrix $G$ and
  projection radius $R$, and may be bounded as in \Cref{fact:sgd:basic:1}.
\end{theorem}
The proof is very similar to that of \Cref{fact:sgd:basic:1}.
Comparing this bound to \Cref{fact:sgd:basic:1}, nearly everything is the same, including the
general growth of the leading term $\kappa^2$ in response to choosing $G_{\id}$ or $\spoly{r}$.
However, as mentioned above,  what differs is the left-hand term: progress is measured against
a time-varying target rather than a time-independent target as in the stochastic model.
As the optimal pricing vector $\bar w^t$ is a function of time, it need not converge in any way, and we cannot hope for convergence in prices either. To illustrate this concretely, we have the following result.

\begin{proposition}
  \label[proposition]{fact:activity_rule:0}
  Suppose the setting of \Cref{fact:sgd:basic:2} with $n = 2$ bidders
  and $m=1$ item, but with step sizes $(\eta^t)_{t\geq 1}$ being any positive reals
  satisfying $\sum_{t\geq 1} \eta^t = \infty$.
  Under bundle or polynomial prices (of any degree),
  there exists a bidding sequence $(b^t)_{t \geq 1}$ consistent with a
  valuation sequence $(v^t)_{t \geq 1} $ such that every sequence of
  optimal price parameters $(\bar w^t)_{t\geq 1}$ fails to converge.
\end{proposition}
\noindent
The proof of this fact, given in Appendix~\ref{sec:proofs}, constructs a concrete bidding sequence whereby
the corresponding optima $(\bar w^t)_{t \geq 1}$ oscillate between two cluster points.

To link the behavior of bidders across rounds and recover price
convergence, the auction can make use of \emph{activity rules}. In
fact, activity rules are used in practice specifically to disallow
certain kinds of adversarial bidding behaviors that are assumed away
by simple models of best-response agents (e.g., bid parking and
sniping)~\citep{ausubel2014market}. The rule that is most firmly
grounded in theory is the \emph{revealed preference activity rule},
also called GARP (after the \emph{generalized axiom of revealed
  preference})~\citep{ausubel2014practical}. We consider here the
strictest form of the rule which requires exact adherence to the GARP
axiom.
A sequence of bid vectors $(b^t)_{t \geq 1}$, placed in response to
prices $(p^t)_{t \geq 1}$ over rounds, satisfies the \emph{GARP
  activity rule} if for every sequence of distinct rounds $t_1, t_2, \ldots,
t_{k'}$ (not necessarily consecutive or ordered),
\begin{equation}\label{eq:garp-rule}
\sum_{k=1}^{k'} (b^{t_{k+1}} - b^{t_k})^\top p^{t_k} \geq 0,
\end{equation}
with the convention $t_{k'+1} = t_1$. Our analysis can also accommodate more relaxed forms where bidders
are allowed to violate GARP in earlier rounds.
The GARP activity rule can be
enforced efficiently in practice using network flow
algorithms~\citep{vohra2004advanced}.

A sequence of bid vectors satisfies the GARP activity rule if and only
if there is a fixed valuation vector $v$ consistent with the entire
sequence of bids and prices (even though the bidders may not be
explicitly considering such a valuation); for completeness, a proof of
this fact is provided as~\Cref{garp-lemma} in the appendix. Under
these circumstances, \Cref{fact:sgd:basic:2} can be strengthened to
obtain a more meaningful bound.
\begin{theorem}
  \label{fact:sgd:basic:2:garp}
  Consider the setting of \Cref{fact:sgd:basic:2}, and assume the bid
  sequence $(b^t)_{t \geq 1}$ satisfies the GARP activity rule with
  respect to prices $(p^t)_{t\geq 1}$. Then there exists a single value
  vector $v$ that is consistent with $b^t$ under prices $p^t$, for all
  rounds $t$, and moreover
  \[
    \scrD_{\lambda}(\hat w^T; v) - \scrD_{\lambda}(\bar w; v)
    \leq
    \cO \left( \frac {\kappa^2 V \ln T} {\sqrt{T}}
    \right),
  \]
  where
  $\hat w^T = \sum_{t=1}^T \eta^t w^t / \sum_{s=1}^T \eta^s$
  is the averaged iterate, $\bar w$ is a minimum for
  $\scrD_{\lambda}(\cdot\,; v)$,
  and $\kappa$ may be bounded as in
  \Cref{fact:sgd:basic:2}. Moreover, if $\lambda \in (0,1/V]$, then
  \[
    \norm{\hat{p}^T - \bar{p}}_\infty
    \leq
    \cO \left( \frac{ \kappa \|G\|_{2,\infty} }{ \lambda }
    \sqrt {\frac {\ln T}{\sqrt{T}}} \right).
  \]
where $\hat{p}^T = G\hat{w}^T$ and $\bar{p} = G\bar{w}$,
and $\|G\|_{2,\infty}$ may be bounded as in \Cref{fact:sgd:params:1}.
\end{theorem}
\noindent
In simpler terms, the GARP activity rule ensures that bidding across rounds is consistent with at least one fixed valuation profile, and for any such profile the auction converges in objective value. With regularization ($\lambda > 0$), we also obtain price convergence, in contrast to Proposition~\ref{fact:activity_rule:0}.


\section{Conclusion.}
\label{sec:conclusion}

This paper obtained concrete bounds on the rate of convergence of
iterative auctions that correspond to subgradient methods for the
underlying optimization problem over prices. Our setup can accommodate
many different pricing schemes and allows one to analyze and compare
them under a single framework. It also admits bidder behaviors beyond
straightforward best-response bidding. We considered two
generalizations of straightforward bidding: stochastic bidders and
arbitrary bidders, with the latter constrained by activity rules.

The convergence rates obtained under both models are very similar. In
both cases, using a more expressive pricing scheme weakens convergence
guarantees. Bounds are proportional to the degrees of freedom in the
prices, so item price bounds are exponentially better than bundle
price bounds. Our analysis quantifies one side of the trade-off
between convergence rate and ability to clear the market, which
mirrors the bias-variance trade-off familiar from statistics. Our
results suggest that an iterative combinatorial auction should use the
simplest class of prices possible that can clear the market. In this
respect, polynomial prices may prove useful in practice.

The stochastic model of bidders allows for errors in their estimates
of the values for various bundles in each round. This translates into
stochastic and imperfect best-response behavior, following standard
random utility models. Our convergence results show that subgradient
auctions perform effective price discovery even under bidding errors,
and the results are robust to correlations between the bidders'
valuations. With regularization, we also obtain convergence in
individual bundle prices. Under the arbitrary bidding model, prices
can oscillate and fail to converge. This motivates the use of revealed
preference activity rules (GARP), which restore consistency with a
fixed valuation vector. Our analysis draws a connection between the
constraints imposed by the activity rule and convergence of the auction.


There are many ways to extend and build on our analysis. An important
avenue for future work is to derive lower bounds on convergence rates
and thereby achieve a separation between pricing schemes. There are
standard tools for producing lower bounds on subgradient
methods~\citep[e.g.][]{agarwal2009information}, but they might
involve pathological constructions lacking economic meaning.
Another avenue is to study other pricing schemes besides polynomial or
bundle pricing, which could be relevant in specific domains and worthy of
study.
Given their prominence in practice, it would also be worthwhile to
analyze monotone price auctions (ascending or descending) within the
present framework. Monotone price auctions can be obtained by
modifying the projection operation in the subgradient method.

\begin{acks}
The authors would like to thank Jason Hartline for valuable feedback.
We also thank participants at seminars at Duke, the University of
Zurich, and the 2015 AMMA Conference for comments.
\end{acks}

\bibliographystyle{ACM-Reference-Format-Journals}
\bibliography{info-prices}

\appendix
\section*{APPENDIX}
\setcounter{section}{0}

\section{Deferred proofs.}
\label{sec:proofs}

This appendix provides complete proofs for all results in the paper.
Several results make use of stand-alone auxiliary lemmas given in the
next section.

\medskip
\begin{proof}[of Theorems~\ref{thm:alloc-price-dual}
    and~\ref{thm:alloc-price-dual-reg}]
  The general existence of an optimum $\bar w$ (for any
  $\lambda\geq 0$) follows from~\Cref{fact:opt:basic}, taking $\nu$ as
  a dirac measure supported on $v$ and nowhere else. When
  $\lambda >0$, the uniqueness of $\bar w$ holds because the
  objective is strictly convex. Note that $F$ is a polyhedral convex set.
  The second duality then follows by applying
  \Cref{fact:duality:basic} with $r(x) = \lambda \|x\|_2^2 / 2$, which
  gives $r^*(y) = \|y\|_2^2 / (2\lambda)$ by direct computation. The
  first duality also follows from \Cref{fact:duality:basic}, but now
  with $r(x) = 0$, meaning $r^*(y) = \iota_{\{0\}}(y)$ and
  $r^*(G^\top q') = \iota_{\ker(G^\top)}(q')$, which gives the dual
  problem as
  \[
    \max\left\{
      v^\top q'
      : q' \in \cset, q \in F, q' - q \in \ker(G^\top)
    \right\},
  \]
  the result following by collapsing $q$ into the constraint on $q'$.
  The invocation of \Cref{fact:duality:basic} also grants the desired optimality conditions;
  in the case of \Cref{thm:alloc-price-dual-reg}, this directly gives the result after translating $r$ and $r^*$ as above,
  whereas with \Cref{thm:alloc-price-dual}, the provided optimality conditions now give
  \[
    \optq' - \optq \in \ker(G^\top),
    \qquad
    \optq' \in U(G\optw;v),
    \qquad
    \optq \in S(G\optw),
  \]
  where the first can be written $\optq' \in \{\optq\} + \ker(G^\top)$,
  thus once again $\optq$ and the condition $\optq\in S(G\optw)$ can be collapsed in by writing
  $\optq' \in S(G\optw) + \ker(G^\top)$ as desired.
\end{proof}

\medskip
\begin{proof}[of~\Cref{integrality}]
If $G$ is full rank, then $\ker(G^\top) = \{0\}$, and thus the
constraint in the dual objective~\eqref{eq:alloc-price-dual} collapses to
$$
q \in \cset \cap (\ker(G^\top) + F) = \cset \cap F = F.
$$
Recalling~\eqref{def:production-set}, $F$ is a polytope with integer
extreme points corresponding to allocations of whole items. Thus the
dual linear program has an integer optimal solution. The optimality
conditions~\eqref{eq:comp-slack} collapse to~ $\optq \in U(\optp) \cap
S(\optp)$ where $\optp = G\optw$, which shows that the dual optimal
solution leads to supporting prices.
\end{proof}

\medskip
\begin{proof}[of~\Cref{prop:w-bounds}]
  Let $v$ be given as specified, and let $\nu$ be a dirac measure
  supported on $v$ and nowhere else. The general existence of an
  optimum $\bar w$ (for any $\lambda\geq 0$) follows from
  \Cref{fact:opt:basic}. When $\lambda >0$, uniqueness of $\bar w$ is
  a consequence of strict convexity of the objective. The bounds on
  $\bar w$ under various choices of $G$ are provided by
  \Cref{fact:opt:basic:1}, choosing $\nu$ as before.
\end{proof}

\medskip
\begin{proof}[of~\Cref{fact:subgradient}]
  The first statement is a consequence of standard subgradient rules
  \citep[Theorem D.2.2.1, Theorem D.4.1.1, Theorem D.4.2.2]{HULL},
  and the definition of subgradient descent is also standard \citep{Bubeck14}.
  The second statement is from \Cref{fact:opt:basic:1}.
\end{proof}

\medskip
\begin{proof}[of~\Cref{fact:subgradient:representer}]
  The expression for $w^t$ holds in the case $w^1 = 0$,
  and by induction and the definition of $w^{t+1}$ in the mechanism,
  \begin{align*}
    w^{t+1}
    &= \gamma^{t+1}\left(w^t - \eta^t g^t\right)
    \\
    &= \gamma^{t+1}(1 - \lambda \eta^t) \left(
      G^\top
      \sum_{s=1}^{t-1} r^t \eta^s (b^s - q^s) \prod_{j = s+1}^{t-1}
      \gamma^{j+1}(1  - \lambda \eta^j)
    \right)
    + \gamma^{t+1}\eta^t G^\top (b^t - q^t)
    \\
    &= G^\top \sum_{s=1}^{t} \gamma^{t+1} \eta^s (b^s - q^s) \prod_{j
      = s+1}^{t} \gamma^{j+1}(1  - \lambda \eta^j)
  \end{align*}
  as desired.
\end{proof}

\medskip
\begin{proof}[of \Cref{fact:sgd:basic:1}]
  First note that an optimal $\bar w$ exists by \Cref{fact:opt:basic}.
  Let $g^t$ be defined as in~\eqref{eq:subgradient}, and let $\sgf^t$
  be any subgradient of $\scrD_{\lambda}(\cdot; v^t)$ at $w^t$.
  Define the quantities
  \[
    L = \max_{t\in[T]}\left\{ \max\{ \|g^t\|_2, \|\sgf^t\|_2\}\right\}
    \qquad\textup{and}\qquad
    B = \max\left\{ \|\bar w\|_2,  \max_{t\in[T]} \|w^t\|_2\right\}.
  \]
  For convenience, define $f(w) = \Exp_\nu[\scrD_{\lambda}(w, v)]$.
  For every $t \in [T]$,
  by Cauchy-Schwarz and the triangle inequality,
  \[
    (\sgf^t - g^t)^\top (w^t - \bar w)
    \leq \|\sgf^t - g^t\|_2\|w^t - \bar w\|_2
    \leq 4LB.
  \]
  Consequently,
  since $\Exp_\nu[\sgf^t - g^t] = 0$, Azuma-Hoeffding grants,
  with probability at least $1-\delta$,
  \begin{align*}
    \sum_{t=1}^T \hat\eta^t (\sgf^t - g^t)^\top(w^t - \bar w)
    &\leq
    4LB \sqrt{2\ln(1/\delta)\sum_{t=1}^T (\hat\eta^t)^2}
    \leq 4LBV\sqrt{2(1 + \ln(t)) \ln(1/\delta)}
    .
  \end{align*}
  Plugging this into \Cref{fact:sgd:basic:0}
  and simplifying the left hand side via convexity
  gives the result.

  For the bounds on $B$ and $L$, first note that $\bar w$ exists by \Cref{fact:opt:basic}.
  Now consider the case $G = G_{\id}$.
  By \Cref{fact:subgradient},
  \[
    \|g^t\|_2
    \leq \|G^\top (b^t - q^t)\|_2 + \lambda \|w^t\|_2
    \leq \|b^t\|_2  + \|q^t\|_2 + \lambda B
    \leq \sqrt{n} + \sqrt{m} + \lambda B.
  \]
  Similarly, taking $b \in U(p^t;v)$ and $q \in S(p^t)$, where $p^t = Gw^t$, to denote subgradient terms for $v$
  drawn from $\nu$~\citep{HULL,ROC},
  \[
    \|\sgf^t\|_2
    \leq \left\|G^\top \Exp_\nu(b - q)\right\|_2 + \lambda \|w^t\|_2
    \leq \sqrt{n} + \sqrt{m} + \lambda B.
  \]
  Lastly, the form of $B$ and bound $\|\bar w\|_2 \leq B$ are provided by \Cref{fact:opt:basic:1}.

  Now consider the case $G = \spoly{r}$.
  By \Cref{fact:subgradient},
  \[
    \|g^t\|_2
    \leq \|G^\top (b^t - q^t)\|_2 + \lambda \|w^t\|_2
    \leq \|G^\top b^t\|_2  + \|G^\top q^t\|_2 + \lambda B
    \leq m^r \sqrt{n} + m^r + \lambda B.
  \]
  The derivation of the remaining properties is as for $G_{\id}$,
  and the result follows.
\end{proof}

\medskip
\begin{proof}[of~\Cref{fact:sgd:params:1}]
  For convenience, define $f(w) = \Exp_\nu[\scrD_{\lambda}(w, v)]$, and
  let $\sgf$ be a subgradient of $f$ at $w$.
  By first order optimality conditions and strong convexity,
  \[
    \epsilon
    \geq f(w) - f(\bar w)
    \geq (w - \bar w)^\top \sgf + \frac{\lambda}{2} \|w - \bar w\|_2^2
    \geq \frac{\lambda}{2} \|w - \bar w\|_2^2,
  \]
  which gives the first result after rearrangement.
  Next, the definition of $\|G\|_{2,\infty}$ gives
  \[
    \norm{\hat{p}^T - \bar{p}}_\infty
    = \norm{G\hat w^T - G \bar w}_\infty
    \leq \norm{G}_{2,\infty} \norm{\hat w^T - \bar w}_2,
  \]
  which gives the second bound after combining with the preceding bound
  and invoking \Cref{fact:sgd:basic:1} to provide the value of $\epsilon$,
  and using the condition $\lambda \leq 1/V$, whereby $V \leq 1 /\lambda$.
  Lastly, to bound $\|G\|_{2,\infty}$ for $G_{\id}$ and $\spoly{r}$,
  first note that the Cauchy-Schwarz inequality implies
  \[
    \|G\|_{2,\infty} = \max\{ \|G_i(x)\|_2 : i\in[n], x\in X\};
  \]
  consequently, in either case, we only need to count the number of 1s
  in rows of $G$.  For $G_{\id}$, this immediately gives $\|G\|_{2,\infty}\leq 1$.
  For $\spoly{r}$, the bundle $\tilde x$ which contains all $m$ items will correspond
  to a row of 1s; since all rows have 0s and 1s, this bundle attains the maximum norm
  (and does not vary with bidder, so we may consider the first bidder, thus row $G_1(\tilde x)$),
  which gives
  \[
    \|\spoly{r}\|_{2,\infty}
    = \|G_1(\tilde x)\|_2
    = \sqrt{\sum_{k=1}^r \binom{m}{k}}
    \leq \sqrt{m^r}.
  \]
\end{proof}

\medskip
\noindent
\begin{proof}[of~\Cref{fact:sgd:basic:2}]
  An optimal $\bar w$ exists by \Cref{fact:opt:basic},
  with measure $\nu$ chosen to
  be the discrete measure over $(v_t)_{t=1}^T$.
  The rate follows from \Cref{fact:sgd:basic:0}
  with $f^t = \scrD_{\lambda}(\cdot, v^t)$ (whereby $\sgf^t - g^t = 0$).
  The estimates on $L$ and $B$ are as in the proof
  of~\Cref{fact:sgd:basic:1}.
\end{proof}

\begin{proof}[of~\Cref{fact:activity_rule:0}]
  First note that for $m=1$, the price parameter vector is
  one-dimensional (recalling our convention that the empty bundle has
  price of 0) and coincides for bundle and polynomial prices of all
  degrees. Thus we write $p = w$ for the price of the single item.

  This proof will construct a sequence of bids, organized into epochs
  ending at times $t_1,t_2,\ldots$, such that for any $k \geq 1$ we
  have $\|\optw^{t_k} - \optw^{t_{k+1}}\|_2 = |\optw^{t_k} -
  \optw^{t_{k+1}}| \geq 1$; since this
  happens for arbitrarily large choices of $k$, the sequence of optima
  is not a Cauchy sequence (and thus does not converge).
  The bidding behavior will be defined in terms of valuations $v^t$ at
  time $t$, and the bids $b^t$ (which will not require more discussion
  in this construction) are merely any choice which maintains the
  consistency of the valuations.

  The construction is as follows. In every round, both bidders have
  the same values, and moreover assign value 0 to the empty bundle. In
  even epochs, both assign value 1 to the item, whereas in odd epochs
  they assign it value 0.

  The epoch lengths $(t_k)_{k\geq 0}$ will be constructed so that
  the objective functions places more emphasis one selecting the single item
  in even epochs, whereas odd epochs will emphasize selecting no item.
  To this end, collect all the step sizes $\eta^s$ from even and odd epochs into $\etae^s$ and $\etao^s$,
  meaning
  \[
    \etae^s = \sum_{i = 1}^s \eta^i \1[\textup{$i$ in even epoch}],
    \qquad\qquad
    \etao^s = \sum_{i = 1}^s \eta^i \1[\textup{$i$ in odd epoch}].
  \]
  Now define $(t_k)_{k \geq 0}$ inductively as $t_0 = 0$, and thereafter,
  given $(t_j)_{j=1}^k$, define $t_{k+1}$ to be an integer sufficiently
  large so that $\etae^{t_{k+1}} > \etao^{t_{k+1}}$ when $k+1$ is even
  and otherwise $\etao^{t_{k+1}} < \etao^{t_{k+1}}$,
  where the existence of $t_{k+1}$ is guaranteed from positivity of $\eta^s$ and $\sum_{s\geq 1} \eta^s = \infty$.

  Now fix some epoch $k$, and set $\etae = \etae^{t_k}$ and $\etao = \etao^{t_k}$ for convenience.
  The objective function evaluates to
  \begin{eqnarray*}
    \sum_{s=1}^{t_k} \eta^s \scrD_{\lambda}(w; v^s)
    & = & 2\etae\max\{1-p,0\} + 2\etao\max\{-p,0\} + (\etae+\etao)\max\{p,0\} \\
    & \mbox{} & \\
    & = & \begin{cases}
            2\etae - 2(\etae + \etao)p & \quad\mbox{if $p \leq 0$,} \\
            (\etae + \etao)p & \quad\mbox{if $p \geq 1$,} \\
            2\etae + (\etao - \etae)p & \quad\mbox{if $p \in [0,1]$.}
  \end{cases}
  \end{eqnarray*}
  As $\etae + \etao > 0$, we see that in the range $p \leq 0$ the minimum is
  uniquely reached at 0, and in the range $p \geq 1$ the minimum is
  uniquely reached at 1. Thus we restrict our attention to the range
  $p \in [0,1]$. There, we see that if $\etao > \etae$, the unique minimum is
  0, while if $\etao < \etae$, the unique minimum is 1. Thus, at the end of
  epoch $k$, we have $\optw^{t_k} = 0$ if $k$ is odd and
  $\optw^{t_k} = 1$ if $k$ is even. This completes the proof.
\end{proof}

\medskip \noindent

\begin{proof}[of \Cref{fact:sgd:basic:2:garp}]
  The existence of $v$ is granted by~\Cref{garp-lemma}.
  Plugging in the consistent sequence $(v^t)_{t \geq 1}$ with $v^t = v$ into
  \Cref{fact:sgd:basic:2} gives the first inequality after collapsing the left
  hand side via Jensen's inequality as in the proof of \Cref{fact:sgd:basic:1}.
  The second inequality is proved analogously to \Cref{fact:sgd:params:1}.
\end{proof}


\section{Technical lemmas.}
\label{sec:technical-lemmas}

This appendix provides results relating to the optimization problem of
pricing studied throughout. The first result is a helper lemma
bounding indirect utilities. The next is a generic duality result. We
prove duality using a general convex regularizer, which covers the
squared norm regularizer used in the paper. The two subsequent results
study the minimizers of the pricing problem. We also provide a
standard convergence bound for online subgradient descent, a proof
that the GARP activity rule implies that bidding is consistent with a
fixed valuation vector, and lastly bounds on the quantity
$\Exp_\nu(\|v\|_\infty)$ of relevance to the analysis of stochastic
bidding.

\begin{lemma}
  \label[lemma]{fact:u_s_nice}
  Both $u$ and $s$ are convex, closed, polyhedral, and nonnegative.
  Moreover, for any prices $p\in \R^{n\ell}$ and any valuations $v\in \R^{n\ell}$,
  \[
    u(p;v) \geq -\|v\|_\infty -\min_{\substack{i\in [n] \\ x\in X}} p_i(x)
    \qquad\textup{and}\qquad
    s(p) \geq \max_{\substack{i\in [n] \\ x\in X}} p_i(x).
  \]
\end{lemma}
\begin{proof}
  Both $u$ and $s$ are convex, closed, and polyhedral by definition \cite[Chapter 19]{ROC}.
  They are nonnegative by their definition, noting that $0 \in F$ and $0 \in \cset$.
  The remaining inequalities follow similarly by definition:
  letting $q_1 \in \cset$ and $q_2\in F$ denote standard basis vectors so that
  \[
    p^\top q_1
    = \min_{\substack{i\in [n] \\ x\in X}} p_i(x)
    \qquad\textup{and}\qquad
    p^\top q_2 =
    \max_{\substack{i\in [n] \\ x\in X}} p_i(x),
  \]
  then
  \begin{align*}
    u(p;v)
    = \max_{q\in \cset}(v-p)^\top q
    \geq v^\top q_1 - p^\top q_1
    \geq - \|v\|_\infty - p^\top q_1
    \qquad\textup{and}\qquad
    s(p)
    = \max_{q\in F} p^\top q
    \geq p^\top q_2.
  \end{align*}
\end{proof}

\begin{lemma}
  \label[lemma]{fact:duality:basic}
  Let closed convex bounded below $r : \R^d \to \R$,
  matrix $G \in \R^{n\ell \times d}$
  and vector $v \in \R^{n\ell}$ be given.
  Then
  $$
    \inf\left\{ u(Gw;v) + s(Gw) + r(w) : z \in \bR^d \right\} \\
    =
     \max\left\{ v^\top q'
                      - r^*(G^\top (q'-q)) : q' \in \cset,
          q \in F \right\},
  $$
  where $r^*(y) = \sup_w (y^\top w - r(w))$ is the convex (Fenchel) conjugate
  of $r$ \cite[Chapter 12]{ROC}.
  Feasible points $\optw$ and $(\optq, \optq')$ are optimal for the
  primal and dual problems, respectively, if 
  they satisfy
  the conditions
  $G^\top (\optq' - \optq) \in \subdiff r(\optw)$,
  $\optq' \in U(G\optw;v)$,
  and $\optq \in S(G\optw)$.
\end{lemma}
\begin{proof} For convenience, let $\iota_S$ denote the
  indicator function for convex set $S$, defined as $\iota_S(q) = 0$
  for $q \in S$ and $+\infty$ otherwise.
  Additionally, we write $u(Gw) = u(Gw;v)$ since $v\in\R^{n\ell}$ is fixed throughout.
  By \Cref{fact:u_s_nice}, both $u$ and $s$
  are convex, closed, polyhedral, and bounded below by 0.
  Since $u(p) = \sup_{q \in \cset} (v-p)^\top q$, applying standard conjugacy rules~\cite[Theorem 12.3]{ROC},
  $$u^*(q) = \iota_{\cset}(-q) + v^\top q.$$
  Similarly, as $s(p) = \sup_{q \in F} p^\top q$, we have
  $$s^*(q) = \iota_F(q).$$
  Combining these pieces, both $u^*$ and $s^*$ are
  polyhedral~\cite[Theorem 19.2, Corollary 19.2.1, Theorem 19.4]{ROC},
  and thus
  \begin{equation}
    (u + s)^*(q) = \min\left\{u^*(q') + s^*(q - q') : q' \in
                       \bR^{n\ell}\right\} \label{eq:infconv}
  \end{equation}
  where attainment on the right-hand side holds because the conjugate
  is proper, which in turn holds because $u$ and $s$ are bounded below
  and $F$ is nonempty~\cite[Theorem 16.4, Corollary 19.3.4]{ROC}.

  Since $u$, $r$ and $s$ are bounded below and finite everywhere,
  it follows by Fenchel duality~\cite[Theorem 3.3.5, Exercise 3.3.9.f]{borwein2010convex} that
  \begin{equation*}
    \inf\left\{ u(Gw) + s(Gw) + r(w) : w \in \bR^d \right\}
    =  \max\left\{ -(u+s)^*(-q) - r^*(G^\top q) : q \in
      \bR^{n\ell} \right\},
  \end{equation*}
  and moreover that a pair $(\optw, \optq)$ is optimal for the primal and dual problems,
  respectively, iff
  $G^\top \optq \in \subdiff r(\optw)$ and $-\optq \in \subdiff
  (u+s)(G\optw)$.

  To simplify the dual expression,
  note by \eqref{eq:infconv} and other conjugacy relations above
  that
  \begin{eqnarray}
    & & \max\left\{ -(u+s)^*(-q) - r^*(G^\top q) : q \in
      \bR^{n\ell} \right\}
    \notag\\
    & = & \max\left\{ - \min\left\{u^*(q') + s^*(-q - q') : q' \in
                       \bR^{n\ell}\right\} - r^*(G^\top q) : q \in
                   \bR^{n\ell} \right\} \notag\\
    & = & \max\left\{ - \iota_{\cset}(-q') - v^\top q' + \iota_F(-q-q')
                      - r^*(G^\top q) : q,q' \in
          \bR^{n\ell} \right\} \notag\\
    & = & \max\left\{ - v^\top q'
                      - r^*(G^\top q) : q \in \bR^{n\ell}, q' \in -\cset,
          q + q' \in -F \right\}.
          \notag
  \end{eqnarray}
  Combining this with the cosmetic changes of variable $q' \mapsto -q'$ and subsequently $q \mapsto q' - q$ leads to
  \begin{equation*}
    \max\left\{ v^\top q' - r^*(G^\top (q'-q)) : q' \in \cset, q \in F \right\}
  \end{equation*}
  as desired.

  It remains to prove the sufficient conditions for optimality.
  Consequently, suppose $(\optw, (\optq,\optq'))$ are given as in the statement,
  meaning
  $G^\top(\optq' - \optq) \in \partial r(\optw)$,
  $\optq' \in U(G\optw)$, and
  $\optq \in S(G\optw)$.
  In order to show these are optimal, they will be shown to satisfy both the optimality conditions above,
  which provide optimality of $(\optw,\optq)$, and also an optimality condition on the infimal convolution,
  which in turn grants optimality of $\optq'$.
  To proceed with this analysis, it is necessary to first reverse the
  change of variable on these assumed conditions,
  meaning first performing $\optq' - \optq \mapsto \optq$ and then $-\optq' \mapsto \optq'$,
  which means the transformed variables satisfy the conditions
  $G^\top\optq \in \partial r(\optw)$,
  $\optq' \in -U(G\optw)$, and
  $\optq + \optq' \in -S(G\optw)$.

  According to the duality statements above,
  in order for $(\optw,\optq)$ to be optimal, it suffices (as above)
  to show $G^\top \optq \in \subdiff r(\optw)$ and $-\optq \in \subdiff
  (u+s)(G\optw)$;
  the first of these holds by assumption, and to decode the second, note
  by the convexity
  of $\cset$ and $F$ as well as the fact that $u$ and $s$ are maximizations
  over linear functions that
  \begin{align*}
    \partial u(p)
    &= \conv\left(
    \{ -q \in \cset : (v-p)^\top q = u(p) \}
    \right)
    = - \{ q \in \cset : (v-p)^\top q = u(p) \}
    = - U(p),
    \\
    \partial s(p)
    &= \conv\left(
    \{ q \in F : p^\top q = s(p) \}
    \right)
    = \{ q \in F : p^\top q = s(p) \}
    = S(p),
  \end{align*}
  whereby the rule $\partial (u+s) = \partial u + \partial s$ for finite convex functions
  and the assumptions $\optq' \in -U(G\optw)$ and $\optq + \optq' \in -S(G\optw)$
  grant
  \[
    -\optq \in S(G\optw) + \{\optq'\}  \in S(G\optw) - U(G\optw) = \partial (u+s)(G\optw),
  \]
  which establishes the second optimality condition.

  It remains to be shown that $\optq'$ is also optimal.
  For this,
  applying first order sufficient conditions
  to $q' \mapsto u^*(q') + s^*(-\optq-q')$,
  it suffices to show that
  \[
    0
    \in \partial u^*(\optq') + \partial (q' \mapsto s^*(-\optq - q'))(\optq')
    = \partial u^*(\optq') - \partial s^*(-\optq - \optq'),
  \]
  where the second equality used composition rules for subdifferentials \citep[Theorem D.4.2.1]{HULL}.
  This in turn completes the proof, since standard conjugacy rules \cite[Proposition E.1.4.3]{HULL}
  grant $G\optw \in \partial u^*(\optq')$ via $\optq' \in \partial u(G\optw) = -U(G\optw)$
  and
  $G\optw \in \partial s^*(-\optq - \optq')$ via $\optq + \optq \in -\partial s(G\optw) = -S(G\optw)$,
  whereby
  $0 = G\optw - G\optw \in \partial u^*(\optq') - \partial s^*(-\optq - \optq')$.
\end{proof}

\begin{lemma}
  \label[lemma]{fact:opt:basic}
  Let closed convex bounded below $r : \R^d \to \R$,
  matrix $G \in \R^{n\ell \times d}$
  and vector $v \in \R^{n\ell}$ be given.
  Suppose further that $r$ is either constant or has compact level sets.
  Then given any probabilty measure $\nu$ over $v \in \R^{n\ell}$ with $\Exp_\nu[ \|v\|_\infty ]<\infty$,
  the function
  \begin{equation}
    \label{eq:opt:repeated}
    h(w) = \Exp_\nu \left[ u(Gw;v) + s(Gw) + r(w) \right]
  \end{equation}
  attains a minimum.
\end{lemma}
\newcommand{\imGt}{\ker(G)^\perp}

\begin{proof}
   If $r$ has compact level sets, then the
  result follows since the rest of $h$ is bounded below, and thus $h$
  itself has compact level sets and attains a minimum.
  Otherwise, suppose $r$ is equal to some constant $c\in \R$ everywhere;
  in this case, $h$ is invariant over $\ker(G)$ (due to $w$ only appearing as $Gw$ now that $r$ is constant),
  so it is convenient to explictly rule out changes along $\ker(G)$
  and consider the auxiliary function
  \[
   f(w) = h(w) + \iota_{\imGt}(w).
  \]
  We will show that $f$ is 0-coercive,
  and thus has compact level sets and attains a minimum \citep[Proposition B.3.2.4]{HULL}.
  This in turn completes the proof, since a minimum for $f$ is also a minimum of $h$
  as follows.
  For any $w \in \R^{d}$
  consider the decomposition $w = w_\perp + w_\tk$ where $w_\perp \in \imGt$ is the orthogonal
  projection of $w$ onto $\imGt$ and $w_\tk\in \ker(G)$ is the orthogonal projection of $w$ onto
  $\ker(G)$.
  Since $h$ is invariant to $\ker(G)$,
  then $h(w) = h(w_\perp)$.
  But as $h$ and $f$ agree over $\imGt$,
  we in fact have $h(w) = h(w_\perp) = f(w_\perp)$,
  and moreover
  \[
    \inf\left\{ h(w) : w\in\R^d\right\}
    = \inf\left\{ h(w) : w\in \imGt \right\}
    = \inf\left\{ f(w) : w\in \imGt \right\}.
  \]
  Since $f$ is infinite off of $\imGt$, its minimum occurs along $\imGt$,
  and the above equalities grant that this minimum is also a minimum for $h$.

  To prove 0-coercivity of $f$, let $w \in \R^d$ be an
  arbitrary nonzero direction,
  and note that it suffices to consider $w \in \imGt$ (since $f$
  is $+\infty$ in other directions). Let $p = Gw$.
  There are now two cases to consider
  on the sign of $p_{\min} = \min\{p_i(x) : i\in [n], x\in X\}$:
  either $p_{\min} < 0$, or not.
  If $p_{\min} < 0$, then \Cref{fact:u_s_nice} grants
  \begin{eqnarray*}
    \lim_{t\to\infty} \frac {f(tw) - f(0)}{t}
    = \lim_{t\to\infty} \frac{\Exp_\nu \left[
      u(tp; v) \right] +
    s(tp) + c - c}{t}
    \geq \lim_{t\to\infty} \frac {\Exp_\nu\left[-\|v\|_\infty - tp_{\min} \right]}{t}
    = -p_{\min}
    > 0,
  \end{eqnarray*}
  which means $f$ is 0-coercive \citep[Proposition B.3.2.4]{HULL}.
  For the other case, $w \neq 0$ combined with $w \in \imGt$ implies
  that $p = Gw \neq 0$.
  Thus $p_{\min} \geq 0$ means that
  $p_{\max} = \max\{ p_i(x) : i\in[n], x\in X\} > 0$.
  Once again invoking \Cref{fact:u_s_nice},
  \begin{eqnarray*}
    \lim_{t\to\infty} \frac {f(tw) - f(0)}{t}
    = \lim_{t\to\infty} \frac{\Exp_\nu \left[
      u(tp; v) \right] +
    s(tp) + c - c}{t}
    \geq \lim_{t\to\infty} \frac {t p_{\max}}{t}
    = p_{\max} > 0,
  \end{eqnarray*}
  proving 0-coercivity \citep[Proposition B.3.2.4]{HULL}.
\end{proof}

\begin{lemma}
  \label[lemma]{fact:opt:basic:1}
  Consider the setting of \Cref{fact:opt:basic}, providing objects $G$
  and $\nu$, but suppose $r(w) = \lambda \|w\|_2^2/2$ for some
  $\lambda \geq 0$. Let $h$ denote the function
  in~(\ref{eq:opt:repeated}), let $\bar w$ denote any minimizer of $h$
  (as provided by \Cref{fact:opt:basic}), and set
  $\cW_0 = \{ w \in \R^d : h(w) \leq h(0)\}$.
  \begin{itemize}
    \item
      If $G = G_{\id}$,
      then every $w \in \cW_0$ (including $\bar w$) satisfies
      $\|w\|_\infty \leq (n+1)V$
      and
      $\|w\|_2 \leq (n+1)V \sqrt{\ell}$.

    \item
      If $G = \spoly{r}$ for some integer $r\geq 1$,
      then every $w \in \cW_0$ (including $\bar w$) satisfies
      $\|w\|_\infty \leq (n+1)V 2^r$
      and
      $\|w\|_2 \leq (n+1)V m^{r/2} 2^r$.
  \end{itemize}
\end{lemma}
\begin{proof}
  Before specializing the choice of $G$, there are a few general
  properties to note. First we have
  \[
    h(0)
    = \Exp_\nu[u(0;v)] + 0 + 0
    = \Exp_\nu\left[ \max_{q\in \cset} v^\top q \right]
    \leq  \Exp_\nu\left[ n \|v\|_\infty \right]
    = nV,
  \]
  where the bound follows from H\"{o}lder's inequality.
  In particular $w \in \cW_0$ implies $h(\bar w) \leq h(w) \leq h(0) \leq nV$.
  Consequently, for any $w \in \cW_0$, letting $p = Gw$, we have
  from \Cref{fact:u_s_nice} that
  \begin{eqnarray*}
    nV
    \geq h(w)
    \geq \Exp_{\nu}[u(p;v)] + 0 + 0
    \geq -\Exp_\nu[ \|v\|_\infty ] - \min_{i \in [n], x \in X} p_i(x)
  \end{eqnarray*}
  and
  \begin{eqnarray*}
    nV
    \geq h(w)
    \geq 0 + s(p) + 0
    \geq \max_{i \in [n], x \in X} p_i(x),
  \end{eqnarray*}
  which together imply
  \begin{equation}
    \norm{p}_{\infty} = \norm{Gw}_{\infty} \leq (n+1)V.
    \label{eq:opt:basic:1:a}
  \end{equation}
  Now consider the case $G = G_{\id}$.  In this case,
  (\ref{eq:opt:basic:1:a}) directly provides $\|w\|_\infty \leq (n+1)V$ for $w\in\cW_0$,
  and thus $\|w\|_2 \leq (n+1) V \sqrt{\ell}$.

  The remainder of the proof will handle the case $G=\spoly{r}$ for some integer $r\geq 1$,
  with an arbitrary $w\in \cW_0$ as before. The feature space now has
  one component for each bundle of size at most $r$; thus we write
  $w(x)$ for the component of $w$ corresponding to $x \in X$ where
  $|x| \leq r$. Let $x \in X$ be of size at most $r$. We have
  $$
  p(x) = \sum_{x' \subseteq x} w(x').
  $$
  By M\"{o}bius inversion~\citep[][Theorem 1, Example 2]{bender1975applications},
  $$
  w(x) = \sum_{x' \subseteq x} (-1)^{|x \setminus x'|} p(x'),
  $$
  and therefore
  $$
  |w(x)| \leq 2^r (n+1)V.
  $$
  Thus $\norm{w}_\infty \leq (n+1)V 2^r$ and $\norm{w}_2 \leq (n+1)V 2^r
  \sqrt{d}$, where $d$ is the dimension of the feature space under
  $\spoly{r}$. Applying the simple bound $d \leq m^r$ yields the result.
\end{proof}

\medskip
The following convergence result is standard; for other versions, see~\cite{Bubeck14}.

\begin{lemma}
  \label[lemma]{fact:sgd:basic:0}
  Let the following objects be given as specified.
  \begin{itemize}
    \item
      A closed convex (but not necessarily bounded) constraint set $\cW \subseteq \R^d$.
    \item
      Iterates $(w^t)_{t\geq 1}$ and supporting objects $(g^t,\eta^t)_{t\geq 1}$
      with $w^1 \in \cW$ arbitrary, $\eta^t = c / \sqrt{t}$
      for some scalar $c>0$,
      and $w^{t+1} = \Pi_{\cW}(w^t - \eta^t g^t)$ where $g^t\in \R^d$ is arbitrary.
    \item
      A sequence of finite convex functions $(f^t)_{t\geq 1}$,
      where $\sgf^t$ will denote an arbitrary subgradient of $f^t$
      at $w^t \in \R^d$.
  \end{itemize}
  Then for any time horizon $T\geq 1$, any comparator $u \in \cW$, and any $L \geq \sup_{t \in [T]} \|g^t \|_2$,
  \begin{align*}
    &\frac 1 {\sum_{t=1}^T \eta^t} \sum_{t=1}^T \eta^t \left(
          f^t(w^t) - f^t(u)
      \right)
      \\
      &\qquad\qquad\leq
      \frac 1 {c\sqrt{T}} \left(
          \|w^1-u\|_2^2
          + 2\sum_{t=1}^T \eta^{t} (w^t - u)^\top (\sgf^t - g_{t}) +  L^2 c^2\ln(eT)
      \right).
  \end{align*}
\end{lemma}
\begin{proof}
  For any $t\geq 1$,
  by properties of orthogonal projection onto closed convex sets \citep[Proposition 3.1.3]{HULL},
  \begin{align*}
    \| w^{t+1} - u \|_2^2
    &= \left\| \Pi_\cW(w^t - \eta^{t} g^{t}) - \Pi_\cW (u) \right\|_2^2
    \\
    &\leq \left\| w^t - \eta^{t} g^{t} - u \right\|_2^2
    \\
    &= \| w^t - u \|_2^2 - 2\eta^{t} (w^t - u)^\top g^t + (\eta^{t})^2 \|g^{t}\|_2^2
    \\
    &\leq \| w^t - u \|_2^2 - 2\eta^{t} (w^t - u)^\top\sgf^t + 2\eta^{t} (w^t - u)^\top(\sgf^t - g^{t}) + (\eta^{t})^2 L^2,
  \end{align*}
  which by rearrangement and the definition of $\sgf^t$ gives
  \begin{align*}
    2\eta^{t}\left(f^t(w^t) - f^t(u)\right)
    &\leq -2\eta^{t}(u - w^t)^\top \sgf^t
    \\
    &\leq
    \| w^t - u \|_2^2 - \| w^{t+1} - u \|_2^2
    + 2\eta^{t} (w^t - u)^\top(\sgf^t - g^t)+ (\eta^{t})^2 L^2.
  \end{align*}
  Summing across $t\in [T]$,
  \begin{align*}
    &\frac 1 {\sum_{t=1}^T \eta^t} \sum_{t=1}^T \eta^t\left(
      f^t(w^t) - f^t(u)
    \right)
    \\
    &\qquad\leq
    \frac 1 {2\sum_{t=1}^T \eta^t} \left(
      \|w^1 - u\|_2^2 - \|w^{T+1} - u\|_2^2
      + 2\sum_{t=1}^T \eta^{t} (w^t - u)^\top(\sgf^t - g^{t}) +  L^2 \sum_{t=1}^T(\eta^{t})^2
    \right),
  \end{align*}
  and the result follows from the estimates
  \[
    \sum_{t=1}^T \eta^t
    \geq
    c \int_1^{T+1}\frac{dx}{\sqrt{x}}
    \geq 2c (\sqrt{T+1} - 1)
    \qquad
    \textup{and}
    \qquad
    \sum_{t=1}^T (\eta^{t})^2
    \leq
    c^2(1 + \ln(T)),
  \]
  as well as the elementary inequality $4(\sqrt{T+1} - 1) \geq
  \sqrt{T}$.
\end{proof}

\medskip
\noindent
The following result is closely related to Afriat's theorem, for which
several proofs are available~\cite{fostel2004two}.
\begin{lemma}\label[lemma]{garp-lemma}
  Let $(b^t)_{t\geq 1}$ be the sequence of bids and $(p^t)_{t\geq 1}$
  be the sequence of prices. There exists a single value vector $v$
  that is consistent with $b^t$ under prices $p^t$, for all
  $t \in [T]$, if and only if the sequence of bids satisfies the GARP
  activity rule~\eqref{eq:garp-rule} with respect to the sequence of
  prices.
\end{lemma}
\begin{proof}
  Given the sequences of bids and prices, form the following system of
  linear inequalities in the variables $\zeta_i^t$ for each $i \in
  [n]$ and $t \in [T]$:
  \begin{equation}\label{eq:afriat-system}
    \zeta_i^t - p_i^{t \top} b_i^t \geq \zeta_i^s - p_i^{t \top} b_i^s
    \hspace{10pt} (i \in [n], s,t \in [T]).
  \end{equation}
  Feasibility of these inequalities is a necessary condition for there
  to exist a valuation $v$ consistent with each bid, because they must
  hold for $\zeta_i^t = v_i^{\top} b_i^t$. To establish sufficiency,
  define
\begin{equation}\label{eq:consistent-v}
  v_i(x) = \min_{t \in [T]} \left\{ \zeta_i^t - p_i^{t \top} b_i^t +
    p^t_i(x) \right\}
\end{equation}
%
for $i \in [n]$ and $x \in X$. If $x^t$ is the bundle associated with
$b_i^t$ (recall that bid vectors in the auction are integer), then
$v_i^{\top} b_i^t = v_i(x^t) = \zeta_i^t$ because, by~\eqref{eq:afriat-system},
$$
\zeta_i^t - p_i^{t \top} b_i^t + p_i(x^t) = \zeta_i^t \leq \min_{s \in
  [T]} \left\{ \zeta_i^s - p_i^{s \top} b_i^s +
    p^s_i(x^t) \right\}.
$$
Now let $e_j$ for $j=1,\dots,\ell$ be the
unit vectors in $\bR^\ell$ and let $e_0$ be the origin. Consider any
$b_i \in H_i$, which can be written as a convex combination
$b_i = \sum_{j=0}^\ell \alpha_j e_j$ for non-negative weight
$\alpha_j$ that sum to 1. We have
$$
  (v_i - p_i^t)^\top (b_i^t - b_i) = \sum_{j=0}^\ell \alpha_j
                                         \left[ (v_i - p_i^t)^\top
                                         (b_i^t - e_j) \right]
  = \sum_{j=0}^\ell \alpha_j \left[
        \zeta_i^t - p_i^{t \top}
        b_i^t + p_i^{t \top} e_j - v_i^\top e_j
        \right]
\geq 0,
$$
where the last inequality follows from~\eqref{eq:consistent-v}.
Therefore the $v$ defined in~\eqref{eq:consistent-v} is consistent
with all bids in the sequence at the given prices, and there exists a
single value vector $v$ consistent with all bids if and only
if~\eqref{eq:afriat-system} is feasible.

By the Farkas lemma, inequalities~\eqref{eq:afriat-system} are
feasible if and only if the optimal value of the following linear
program is 0. The LP has a non-negative variable $\lambda_{st}$ for
each $s,t \in [T]$.
\begin{eqnarray*}
  \min_{\lambda \geq 0} & \ds \sum_{s,t \in [T]} (b^t - b^s)^\top p^s
                          \,\lambda_{st} & \\
\mbox{s.t.} & \ds \sum_{s \in [T]} \lambda_{st} + \sum_{s \in [T]}
                \lambda_{ts} = 0 & (t \in [T]). \\
\end{eqnarray*}
This is exactly the LP corresponding to a minimum cost circulation
problem over a complete directed graph with a node for each
$t \in [T]$, where the cost of edge $(s,t)$ is $(b^t - b^s)^\top p^s$.
See~\cite[Chapter 4]{bertsekas1998network} for the equivalence of
min-cost flow and circulation problems to their LP formulations.
If~\eqref{eq:garp-rule} does not hold, then there is a negative cost
cycle, and the value of the LP is negative (in fact, unbounded below).
Conversely, assume~\eqref{eq:garp-rule} holds. By the conformal
realization theorem~\cite[Proposition 1.1]{bertsekas1998network}, any
circulation can be decomposed into a sum of simple cycle flows.
As~\eqref{eq:garp-rule} implies that the cost of each simple cycle is
non-negative, the cost of the circulation itself is non-negative, by
linearity of cost. The optimal value of the LP is therefore 0, which
is achieved by setting each $\lambda_{st}$ to 0.
Thus~\eqref{eq:afriat-system} is feasible if and only if the bids
satisfy the GARP activity rule.
\end{proof}

\medskip
\noindent
The following result bounds the maximum value, across all agents and
bundles, under the standard logit (Gumbel) and probit (Gaussian) random utility models.
The result also covers the more general case of any subgaussian error
term, and we stress that error components may be correlated in this case.

\begin{lemma}
  \label[lemma]{fact:Vbounds}
  Consider the following choices of distribution $\nu$ over
  $v\in\bR^{n\ell}$, which takes the form $v = \meanv + \epsilon$.
  In each case $\meanv \in \bR^{n\ell}$ and $\sigma\in\bR^{n\ell}$ are deterministic quantities,
  with $\sigma_{\max} = \max_i \sigma_i$ for convenience, and
  $\epsilon \in \bR^{n\ell}$ has zero mean.
  \begin{itemize}

    \item 
      If $\epsilon_i$ is drawn from a Gumbel distribution with scale
      parameter $\sigma_i$,
      then $$\Exp_\nu [ \|v\|_\infty ] \leq \|\meanv\|_\infty + 2\sigma_{\max}\ln(2n\ell\sqrt{\pi}).$$

    \item
      If $\epsilon_i$ is drawn from a Gaussian with mean $\meanv$ and
      variance $\sigma_i^2$, then
      $$\Exp_\nu [ \|v\|_\infty ] \leq \|\meanv\|_\infty + \sigma_{\max}\sqrt{2\ln(2n\ell)}.$$

    \item 
      If $\epsilon_i$ is subgaussian with parameters
      $(0, \sigma_i^2)$,
      then $$\Exp_\nu [ \|v\|_\infty ] \leq \|\meanv\|_\infty + \sigma_{\max}\sqrt{2\ln(2n\ell)}.$$
  \end{itemize}
\end{lemma}
\begin{proof}
  First recall the following standard derivation linking maxima of random variables
  and their moment generating functions \citep[e.g.,][Section 2.5]{blm_conc}.  For any
  (possibly dependent) random variables $(X_1,\ldots,X_M)$, note by convexity
  for any $t\geq 0$ that
  \begin{align}
    \exp\left(t \Exp_\nu[\max_i |X_i| ] \right)
    &\leq
    \Exp_\nu\left[ \exp( t\max_i |X_i| ) \right]
    \notag\\
    &=
    \Exp_\nu\left[\max_i \exp( t |X_i| ) \right]
    \notag\\
    &\leq
    \Exp_\nu\left[\sum_i \exp( t |X_i| ) \right]
    \notag\\
    &\leq
    \sum_i \left( \Exp_\nu\left[ \exp( t X_i ) \right] + \Exp_\nu\left[\exp(-tX_i)\right] \right).
    \label{eq:ranutil:helper}
  \end{align}

  \noindent
  For the Gumbel distributions with scale parameter $\sigma_i$, first note for every $s\in\bR$ that
  $\Exp_\nu[\exp(s\epsilon_i)] = \Gamma(1-s\sigma_i)$.  Combining this
  with~\eqref{eq:ranutil:helper}, we have
  \[
    \Exp_\nu[\max_i |\epsilon_i| ]
    \leq \frac 1 t \ln\left(
      \sum_i
      \left(\Gamma(1-t\sigma_i) + \Gamma(1+t\sigma_i)\right)
    \right),
  \]
  whereby the choice $t = 1/(2\sigma_{\max})$ combined with properties of $\Gamma$ grants
  \[
    \Exp_\nu[ \max_i |\epsilon_i| ]
    \leq 2\sigma_{\max} \ln\left(
      \sum_i
      \left(\Gamma(1-\sigma_i /(2\sigma_{\max})) + \Gamma(1+\sigma_i/(2\sigma_{\max}))\right)
    \right)
    \leq
    2\sigma \ln\left(2n\ell\sqrt{\pi})\right),
  \]
  and the desired bound follows since $\Exp_\nu[\|v\|_\infty] \leq \|\meanv\|_\infty + \Exp_\nu[\|\epsilon\|_\infty]$.

  Next consider the bound on arbitrary subgaussian random variables, which will immediately grant the Gaussian bound.
  As each $\epsilon_i$ satisfies $\Exp_\nu[ \exp(t\epsilon_i) ] \leq \exp(t^2\sigma_i^2/2)$ by definition,~\eqref{eq:ranutil:helper} gives
  \begin{align*}
    \exp\left(t \Exp_\nu[ \max_i |\epsilon_i| ] \right)
    &\leq
    \sum_i \left( \Exp_\nu\left[ \exp( t\epsilon_i) \right] + \Exp_\nu\left[\exp(-t\epsilon_i )\right] \right)
    \\
    &\leq
    2 \sum_i \exp(t^2\sigma_i^2/2)
    \\
    &\leq
    2 n\ell \exp(t^2\sigma_{\max}^2/2).
  \end{align*}
  We therefore obtain
  \begin{align*}
    \Exp_\nu[\max_i |\epsilon_i| ]
    &\leq
    \frac {\ln(2n\ell)} t
    + \frac {t\sigma_{\max}^2}{2}.
  \end{align*}
  Since this expression holds for all $t>0$, using the optimal choice
  $t = \sqrt{2\ln(2n\ell)/\sigma_{\max}^2}$ gives the result
  since $\Exp_\nu[\|v\|_\infty] \leq \|\meanv\|_\infty + \Exp_\nu[\|\epsilon\|_\infty]$.
%
\end{proof}

\end{document}